\documentclass[%
 reprint,
nofootinbib,
 amsmath,amssymb,
 aps,
pra,
]{revtex4-2}
\newcommand{\hi}{\mathcal{H}}
\newcommand{\hione}{\mathcal{H}_1}
\newcommand{\hitwo}{\mathcal{H}_2}

\usepackage{comment}

\newcommand\indep{\protect\mathpalette{\protect\independenT}{\perp}}
\def\independenT#1#2{\mathrel{\rlap{$#1#2$}\mkern2mu{#1#2}}}

\newcommand{\M}{\mathsf{M}}
\newcommand{\R}{\mathcal{R}}

\newcommand{\Bor}{\text{Bor}}

\def\state{\mathscr{D}(\mathcal{H})}

\newcommand{\bh}{\mathcal{B}(\hi)}

\newcommand{\eh}{\mathcal{E(H)}} 

\newcommand{\E}{\mathsf{E}}

\usepackage[page,titletoc,title]{appendix}

\usepackage{physics}
\usepackage{soul}
\usepackage{relsize}
\usepackage{lmodern}
\usepackage{slantsc}
\usepackage{bbm}
\usepackage{graphicx}
\usepackage{amsmath,amsthm}
\usepackage{bbold}
\usepackage{amsfonts}
\usepackage{amssymb}
\usepackage{amsthm}
\usepackage{amsbsy}
\usepackage{mathrsfs}
\usepackage{varioref}
\usepackage{dsfont}
\usepackage{bm}
\usepackage{color}
\usepackage[nopar]{lipsum}
\usepackage{lipsum}
\usepackage{tikz-cd}
\usepackage{color}
\usepackage[T1]{fontenc}
\usepackage{xcolor}
\usepackage{csquotes}
\usepackage{mathtools}
\usepackage{tikz-cd}
\usepackage{slashed}
\usepackage{hyperref}

\usepackage[compat=1.1.0]{tikz-feynman}
\usepackage{tikz-cd}
\usetikzlibrary{arrows}
\usepackage{pgfplots}
\pgfplotsset{compat=1.15}
\usepackage[bottom]{footmisc}

\newtheorem{theorem}{Theorem} 
\newtheorem{lemma}{Lemma}     
\newtheorem{corollary}{Corollary}
\newtheorem{proposition}{Proposition}
\newtheorem{definition}{Definition}

\newcommand{\U}{\mathcal{U}} 

\newcommand{\V}{\mathcal{V}}
\newcommand{\W}{\mathcal{W}}

\tikzset{every picture/.style={line width=0.75pt}} 

\renewcommand{\S}{\mathcal{S}}

\newcommand{\sam}[1]{{\color{black}{#1}}}

\begin{document}

\title{Einstein causality of quantum measurements in the Tomonaga-Schwinger picture}
\author{Samuel Fedida}
\email{sylf2@cam.ac.uk}
\affiliation{Centre for Quantum Information and Foundations, DAMTP, Centre for Mathematical Sciences, University of Cambridge, Wilberforce Road, Cambridge CB3 0WA, UK}
\date{\today}

\begin{abstract}
    We \sam{investigate} a generalisation to Lüders' rule à la Aharonov-Albert in those globally hyperbolic spacetimes which allow unitarily equivalent Hilbert spaces to be defined along Cauchy hypersurfaces, thus relying on the existence of an interaction picture à la Tomonaga-Schwinger. We show that under this rule \sam{and under the additional assumptions of the integrability and unitarity of the Tomonaga-Schwinger dynamics and the foliation-independence of rays on acausal Cauchy hypersurfaces}, selective quantum measurements satisfy a state-independent anyonic commutation relation over spacelike-separated precompact regions. We highlight that this propagates to positive operator-valued measures, where the commutation is necessarily bosonic. In the \sam{instantaneous-measurement idealisation}, this implies quantum no-signalling for non-selective measurements. We then examine Sorkin's impossible measurements and show that immediate contradictions can be averted as long as collapse-inducing measurements are irreversible. \sam{These results reaffirm the consistency of the Tomonaga-Schwinger picture of relativistic quantum theory, for which unitarity, integrability and foliation-independence of the states exclude superluminal signalling despite the ``instantaneity" of a side-cone measurement collapse rule}. We finish by discussing the possibility of extending such results beyond the interaction picture.
\end{abstract}

\maketitle

\section{Introduction}

The concept of quantum measurement, though a cornerstone in the study of quantum mechanics, has long been a source of conceptual tension. Indeed, while the standard nonrelativistic collapse postulate (sometimes called Lüders' rule \cite{Luders1950}) provides a pragmatic approach to calculating measurement outcomes, it introduces an abrupt and discontinuous change to the system's evolution, seemingly at odds with the smooth, unitary evolution given by the Schrödinger equation. This issue becomes particularly acute when considering relativistic scenarios, where the notion of \enquote{instantaneous} measurements and collapse clash with the classical intuitions of locality and causality. In the setting of relativistic quantum theory, such principles have several different implementations which can be inequivalent \cite{kent_nonlinearity_2005}. \\

Einstein causality is one realisation of the relativistic no-superluminal-signalling principle within quantum theory. It relies on the notion of (relative) compatibility \cite{hardegree_relative_1977,pulmannova_relative_1980,lahti_coexistence_2003}, or mutual commutativity, of observables, which has a very natural physical meaning. Indeed, let $A \in \bh$ and $B \in \bh$ be two observables with discrete spectra, and $\ket{\psi} \in \hi$. The projection postulate of standard quantum theory says that after a measurement involving $A$ with outcome $i$, the resulting state of the system is $A_i \ket{\psi}/\norm{A_i \ket{\psi}}$ where $A_i$ is a projection corresponding to the outcome $i$ in the spectrum of $A$. A subsequent measurement involving $B$ with outcome $j$ will result in the state $B_j A_i\ket{\psi}/\norm{B_j A_i \ket{\psi}}$. If the order of measurements is reversed, then the final state is instead $A_i B_j \ket{\psi}/\norm{A_i B_j \ket{\psi}}$. However, if $A$ and $B$ are \emph{compatible} relative to $\ket{\psi}$, i.e. provided $A_iB_j\ket{\psi} = B_j A_i\ket{\psi}$, then the order of the successive measurements is irrelevant, i.e. these are \enquote{simultaneously} measurable. \\

In a relativistic context, the notion of \enquote{order of measurements} is meaningless over spacelike-separated regions: there is no intrinsic plane of simultaneity in Minkowski spacetimes -- nor in \enquote{nicely behaved} generalisations of those, namely globally hyperbolic spacetimes, which are those spacetimes which are commonly thought to be \enquote{well-behaved} in terms of causality and representative of the physical world we live in. Thus, the assumption that any two observables from local algebras associated with two spacelike separated regions of spacetime mutually commute, called the Einstein causality postulate, is commonly taken as an axiom in relativistic quantum field theory (QFT) \cite{Busch2009}. \\

For example, algebraic QFT \cite{haag_algebraic_1964,fewster_algebraic_2019} (AQFT) is a powerful framework to formulate gauge field theories in a relativistically covariant and mathematically sound fashion which explicitly postulates this principle. It notably allows to discuss the notion of quantum fields and algebras of local observables beyond the interaction picture of textbook \enquote{physicists'} QFT, which separates a free and interaction Hamiltonian in the system's total Hamiltonian. Notwithstanding the centrality of the interaction picture of QFT in high-energy physics, it is well-understood to be formally constrained to certain classes of QFTs and spacetimes. Indeed, it assumes the unitary equivalence of the Hilbert spaces across several different regions of spacetime. This is known to be impossible for generic QFTs in curved spacetimes -- it can be shown \cite{Torre_1999} that even for free scalar QFTs, there exists globally hyperbolic spacetimes for which such a notion of unitary equivalence across regions of spacetime cannot be implemented. \\

Recent developments have further solidified our understanding of measurements in QFT on curved spacetimes. Several approaches -- including the Hellwig-Kraus light-cone collapse \cite{Hellwig1970,Kraus1971}, the Fewster-Verch measurement scheme of AQFT \cite{fewster_quantum_2020} and detector-based measurement theory \cite{PoloGomez2022,pranzini_detector-based_2025} -- are actively debated in the current literature (see \cite{fraser_note_2023} for a historical discussion comparing and contrasting different approaches to measurement schemes in QFT). Nonetheless, thought experiments keep pushing the boundaries of this quantum principle \cite{Aharonov1981b,Aharonov1984,Sorkin1993}, \sam{with notably a discussion of local and non-local measurements being critical in relativistic settings \cite{breuer_measurements_2007}}, and although challenges can at least partly be addressed \cite{Mould1999, ruep_causality_2022,fewster_measurement_2023} within the above formalisms, the study of measurement schemes in relativistic quantum theory remains an interesting and active field of research. \\

In the following, we denote by $\bh$ the space of bounded linear operators on a Hilbert space $\hi$, and $\eh \subset \bh$ the space of effects on $\hi$: the space of positive bounded linear operators less than or equal to the identity $\mathbb{1}_{\bh}$. We write $\state$ for the set of states in $\hi$, i.e. the subset of positive trace-class operators on $\hi$ with trace $1$, and $U(\hi)$ as the set of all unitary operators on $\hi$. We write $\mathcal{B}(\hione,\hitwo)$ for the space of bounded linear maps between $\hione$ and $\hitwo$. We write $\comm{A}{B}_\phi := AB - e^{i\phi} BA$. \\

Let $(\mathcal{M},g)$ be a d-dimensional globally hyperbolic Lorentzian spacetime\footnote{A spacetime $(\mathcal{M},g)$ is said to be \emph{globally hyperbolic} if it contains no closed causal curve (i.e. it is causal) and if $\forall p,q \in \mathcal{M}$, the double cone $J^-(p) \cap J^+(q)$ is compact \cite{Hounnonkpe_2019}.}, $p \in \mathcal{M}$, $\U \subset \mathcal{M}$.
\begin{enumerate}
    \item We denote by $J^+(p)$ (respectively, $J^-(p)$) the collection of all points in $\mathcal{M}$ which can be reached by a future-directed (respectively, past-directed) causal curve starting from $p$, and $J(p) := J^+(p) \cup J^-(p)$. We write $J^\pm(\U) := \bigcup_{p \in \U} J^\pm(p)$
    \item We denote by $I^+(p)$ (respectively, $I^-(p)$) the collection of all points in $\mathcal{M}$ which can be reached by a future-directed (respectively, past-directed) timelike curve starting from $p$. We write $I^\pm(\U) = \bigcup_{p \in \U} I^\pm(p)$.
    \item By \emph{strict causal future} (respectively \emph{strict causal past}) of $\U$, we mean $J^+(\U) \smallsetminus \U$ (respectively $J^-(\U) \smallsetminus \U$).
    \item We denote by $D^+(\U)$ (respectively, $D^-(\U)$) the future (respectively, past) domain of dependence of $\U$, that is, the set of points $p \in \mathcal{M}$ such that every past-directed (respectively, future-directed) inextendible timelike curve starting at $p$ intersects $\U$. We write $D(\U) := D^+(\U) \cup D^-(\U)$.
    \item We denote by $B^+(\U)$ (respectively $B^-(\U)$) the future (respectively past) boundary of $\U$, i.e. \begin{equation}B^{\pm}(\U) = \{p \in \overline{\U} \mid J^\pm(p) \cap \overline{\U} = \{p\}\} \, .\end{equation}
\end{enumerate}
If $\U, \V \subset \mathcal{M}$, we write $\U \indep \V$ whenever $\U \subset \V'$, i.e. whenever $\U$ and $\V$ are spacelike-separated regions. \\

In the present paper, we will expand on the classical intuitions of the meaninglessness of instantaneous measurements by \sam{exploring} a generalisation to Lüders' rule in globally hyperbolic spacetimes which resembles that of Aharonov-Albert \cite{Aharonov1981b,Aharonov1984,breuer_relativistic_1998,breuer_measurements_2007}, and explore how Einstein causality can be \sam{understood as a consistency requirement} in the setting of an interaction picture. Although arguably naive though natural, the prescription we shall give will be shown to be a good model for many relativistic scenarios. In App. \ref{sec:diff geo}, we remind the reader of useful results of Lorentzian geometry to talk about Einstein causality in globally hyperbolic spacetimes, notably 
\begin{itemize}
    \item Given two precompact and spacelike-separated regions $\U$ and $\V$ of a globally hyperbolic spacetime, one can always have two Cauchy hypersurfaces such that one ($\Sigma_f$) is in the future of both regions while the other ($\Sigma_i$) is in the past of both regions;
    \item One can always find two foliations $\mathcal{S}_1$ and $\mathcal{S}_2$ of the spacetime which both contain $\Sigma_i$ and $\Sigma_f$, in which in $\mathcal{S}_1$ one \enquote{sees} $\U$ as \enquote{happening before} $\V$, whereas in $\mathcal{S}_2$ one \enquote{sees} $\U$ as \enquote{happening after} $\V$;
    \item The future and past boundaries of a precompact region of spacetime are covariant with respect to isometries: a point in the future and past boundaries of a region stays in the future and past boundaries, respectively, of that region upon a change of coordinates.
\end{itemize}
In section \ref{sec:Eins caus}, we write down a state update rule in globally hyperbolic spacetimes from which the Einstein causality of quantum measurements conducted across spacelike-separated (pre)compact regions follows \sam{under the assumptions of the integrability and unitarity of the dynamics and the foliation independence of the rays}. Such results are then showed to extend to positive operator-valued measures (POVMs), which assign a probability to the most general type of measurements undertaken in quantum theory, where the phase is now necessarily bosonic. In section \ref{sec:non-selective measurements}, we examine the case of non-selective measurements and explore how the Einstein causality for selective measurements imply the quantum no-signalling of non-selective measurements in those simple cases where the measurements are assumed to be instantaneous. 
\\ 

In section \ref{sec:impossible}, we explore how Sorkin's impossible measurements can be understood in the framework explored in this paper. In particular, we show that we can avoid immediate contradictions by assuming that collapse-inducing measurements are fundamentally irreversible.
Finally, in section \ref{sec: beyond interaction pic}, we discuss how these results compare to more general settings, where one cannot assume the existence of an interaction picture. \\

\section{Einstein causality in the interaction picture}

\label{sec:Eins caus}

\subsection{Measurements over topological spaces} 

Some quantum observables are well known to not be representable as projection-valued measures (PVMs) -- a time observable, a position observable for photons and a phase observable are prime examples \cite{beneduci_note_2013}. Generally, in finite dimensions, quantum measurements are implemented by a collection of \sam{(Kraus)} operators $\{M_n\}_n \subset \bh$ satisfying $\sum_n M_n^\dagger M_n = \mathbb{1}_{\bh}$ such that, given a pure state $\ket{\varphi} \in \hi$ before the measurement, the probability that outcome $n$ occurs is given by the Born rule $p(n) = \expval{M_n^\dagger M_n}{\varphi} = \norm{M_n \ket{\varphi}}^2$ and the (normalised) state of the system after this operation is given by $\frac{M_n \ket{\varphi}}{\norm{M_n \ket{\varphi}}}$. Equivalently, we can define the family -- a positive operator-valued measure (a POVM) $\{E_n\}_n$ of positive operators $E_n = M^\dagger_n M_n$ with $M_n = \sqrt{E_n}$, $p(n) = \expval{E_n}{\varphi} = \norm{\sqrt{E_n} \ket{\varphi}}^2$ and post-measurement state $\frac{\sqrt{E_n} \ket{\varphi}}{\norm{\sqrt{E_n} \ket{\varphi}}}$. A PVM is a special type of POVM. More generally, POVMs can be defined over topological spaces as follows \cite{beneduci_note_2013,Yashin_2020}.

\begin{definition}
    \label{def: POVM interaction picture}
    Let $X$ be a topological space and $\Bor(X)$ the Borel $\sigma$-algebra on $X$. A \emph{POVM} is a map $\E : \Bor(X) \to \eh$ such that
    \begin{equation}
        \E\Bigg(\bigcup_{n=1}^\infty \Delta_n \Bigg) = \sum_{n=1}^\infty \E(\Delta_n)
    \end{equation}
    where $\{\Delta_n\}_{n}$ is a countable family of disjoint sets in $\Bor(X)$ and the series converges in the weak operator topology. We consider normalised POVMs, for which $\E(X) = \mathbb{1}$. 
\end{definition}

The condition of countable additivity precisely reflects the condition for the POVM to induce a probability measure: for any $\omega \in \state$, $p_\omega^{\E}(\cdot) := \Tr[\omega \E(\cdot)]$ is indeed a probability measure. Thus, since measurements can be defined for any POVM (and vice-versa), the notion of measurements on topological spaces also inherits such additivity properties.

\begin{definition}
    \label{def: mmt interaction picture}
    Let $X$ be a topological space and $\Bor(X)$ the Borel $\sigma$-algebra on $X$. A \emph{measurement} is a map $\M : \Bor(X) \to \bh$ such that
    \begin{equation}
        \M\Bigg(\bigcup_{n=1}^\infty \Delta_n \Bigg)^\dagger \M\Bigg(\bigcup_{n=1}^\infty \Delta_n \Bigg) = \sum_{n=1}^\infty \M(\Delta_n)^\dagger \M(\Delta_n)
    \end{equation}
    where $\{\Delta_n\}_{n}$ is a countable family of disjoint sets in $\Bor(X)$ and the series converges in the weak operator topology, and such that
    \begin{equation}
        \M(X)^\dagger \M(X) = \mathbb{1}_{\bh} \, .
    \end{equation}
\end{definition}

The above generalises the identity $\E_n = \M_n^\dagger \M_n$ to the countable case $\E(\Delta_n) = \M(\Delta_n)^\dagger \M(\Delta_n)$. Conversely, as in the finite case, $\sqrt{\E} : \Delta \mapsto \sqrt{\E}(\Delta) := \sqrt{\E(\Delta)}$ (which exists as $\E$ is positive semi-definite) is a measurement.

\begin{proposition}
    \label{prop: U1 M U2 is a mmt}
    Let $M : \Bor(X) \to \bh$ be a measurement, $U_1,U_2 \in U(\hi)$ be two unitaries. Then the map $U_1 M(\cdot) U_2 : \Bor(X) \to \bh$ is a measurement.
\end{proposition}

\begin{proof}
    For all countable families of disjoint sets $\{\Delta_n\}_{n}$ in $\Bor(X)$,
    \begin{widetext}
        \begin{align}
        \Bigg(U_1 \M\Big(\bigcup_{n=1}^\infty \Delta_n \Big) U_2 \Bigg)^\dagger \Bigg(U_1 \M\Big(\bigcup_{n=1}^\infty \Delta_n \Big) U_2 \Bigg) &= U_2^\dagger \M\Big(\bigcup_{n=1}^\infty \Delta_n \Big)^\dagger \M\Big(\bigcup_{n=1}^\infty \Delta_n \Big) U_2 \\ &= U_2^\dagger\sum_{n=1}^\infty \M(\Delta_n)^\dagger \M(\Delta_n) U_2 \\ &= \sum_{n=1}^\infty U_2^\dagger \M(\Delta_n)^\dagger U_1^\dagger U_1 \M(\Delta_n) U_2 \\
        &= \sum_{n=1}^\infty \Big(U_1 \M(\Delta_n) U_2\Big)^\dagger U_1 \M(\Delta_n) U_2
    \end{align}
    \end{widetext}
    in the weak operator topology, and
    \begin{align}
        (U_1 \M(X) U_2)^\dagger (U_1 \M(X) U_2) &= U_2^\dagger \M(X)^\dagger U_1^\dagger U_1 \M(X) U_2 \\ &= U_2^\dagger \M(X)^\dagger \M(X) U_2 \\ &= U_2^\dagger U_2 = \mathbb{1}_{\bh} \, .
    \end{align}
\end{proof}

In spacetimes where there is no notion of simultaneity for spacelike-separated observers, such as Minkowski spacetime and, more generally, globally hyperbolic Lorentzian spacetimes, one should expect pairs of POVMs evaluated over spacelike separated regions to satisfy some notion of commutativity: there is no unique \enquote{plane of simultaneity} in such spacetimes, so the \enquote{order} of the successive measurements carried out in spacelike-separated regions ought to be irrelevant. The commutativity of local observables for spacelike-separated regions is then sometimes called the \emph{Einstein causality condition}, imposed as a physical postulate in most of the scenarios considered in relativistic quantum physics. \\

\sam{Yet another layer of complexity that one can add to the picture above is to consider ensembles of measurement operators conducted in precompact regions. In a given precompact region $\U \subset \mathcal{M}$, one can consider state measurements through $\{\M_i(\U)\}_{i = 1}^{n \in \mathbb{N}}$ such that $\sum_{i=1}^n \E_i(\U) = \sum_{i=1}^n \M_i(\U)^\dagger \M_i(\U) = \mathbb{1}_{\bh}$ for all $\U \subset \mathcal{M}$. Each $\M_i(\U)$ is to be understood as an operation in $\U$, with outcome probability $p_\U(i)$ given by the Born rule. Let us examine how this notion of measurements fits in with relativistic quantum dynamics.} \\

\subsection{Quantum measurements in the Tomonaga-Schwinger picture} 

\label{sec:tomonaga-schwinger mmts}

In theoretical physics, it is common practice to associate quantum states (often taken to be pure, i.e. rays in Hilbert space) to Cauchy slices of globally hyperbolic spacetimes as $\Sigma \mapsto \ket{\psi[\Sigma]} \in \hi$ and evolving them e.g. through the Tomonaga-Schwinger equation
\begin{equation}
    i \hbar \frac{\delta\ket{\psi[\Sigma]}}{\delta \Sigma(x)} = H_I(x) \ket{\psi[\Sigma]} \, .
\end{equation}
where $H_I$ is some (interaction) Hamiltonian \sam{density}. More generally, one ought to associate density operators to Cauchy slices to account for mixtures; however, it will be sufficient for us to consider the action of POVMs on pure states since operators are fully determined by how they act on pure states. \\

\sam{A crucial consistency condition for what follows is that of \emph{integrability} \cite{koba_integrability_1950,blum_perturbative_2025}. This notion is related to the fact that ``pushing the surface forward" in the Tomonaga-Schwinger equation does not depend on the path of deformation. Formally, it is the condition that the functional derivatives commute:}
\begin{equation}
    \label{eqn:integrability}
    \frac{\delta}{\delta \Sigma(x)} \frac{\delta}{\delta \Sigma(y)} \Psi[\Sigma] \stackrel{!}{=} \frac{\delta}{\delta \Sigma(y)} \frac{\delta}{\delta \Sigma(x)} \Psi[\Sigma] \, .
\end{equation}
\sam{In turn, assuming the unitary equivalence of the Hilbert spaces across spacetime, this is ensured by the requirement that the interaction Hamiltonian commutes over spacelike separations, i.e.}
\begin{equation}
    \comm{H_I(x)}{H_I(y)} \stackrel{!}{=} 0 \qquad \forall x \indep y \, ,
\end{equation}
\sam{to be understood in the distributional sense. This is a microcausality condition at the level of the interaction Hamiltonian density, and is certainly weaker than imposing pointwise microcausality for \emph{all} field operators of the theory. This in turn guarantees that evolving $\Sigma$ to $\Sigma'$ independent of the foliation chosen, and only depends on the endpoints.} \\

Importantly, the validity of the Tomonaga-Schwinger equation in the context of QFT in generic globally hyperbolic spacetimes of dimension greater than $2$ has been disproven by \cite{Torre_1999}, where it was argued that the formalism of AQFT is generally better suited for such a functional evolution of the quantum state. Moreover, Helfer showed that, in the Hadamard representation of the canonical commutation relations, the dynamical evolution cannot be unitarily implemented \cite{Helfer_1996}. Furthermore, the Tomonaga-Schwinger equation is intrinsically defined in the interaction picture -- see \cite{lienert_borns_2020} for a detailed discussion on this matter. These objections are ultimately linked with the inexistence of unitarily equivalent Hilbert spaces across different Cauchy hypersurfaces -- this already holds for free fields on spacetimes with topology $\mathbb{R} \times \mathbb{T}^n$ for $n>1$ \cite{Torre_1999}. \\

On the other hand, the validity of the Tomonaga-Schwinger equation on a flat spacetime in $1+1$ dimensions is shown in \cite{Torre_1998}, where it was proven that dynamical evolution along arbitrary spacelike foliations is indeed unitarily implemented on the same Fock space. This highlights that a map $\Sigma \mapsto \hi_\Sigma \cong \hi$ for all acausal Cauchy hypersurfaces $\Sigma$ of a globally hyperbolic spacetime $(\mathcal{M},g)$ for some fixed $\hi$ does exist for some globally hyperbolic spacetimes and quantum field theories, but does not exist for others. In most of this paper, this restriction will play an important role. \\

We now postulate a state update rule in this Tomonaga-Schwinger picture which, though similar in some ways to the Hellwig-Kraus light-cone collapse \cite{hellwig_pure_1969,Hellwig1970,Kraus1971}, differs from it in certain key aspects. In fact, such a functional approach to state measurements has been approached by Aharonov and Albert in response to Hellwig and Kraus \cite{Aharonov1981b,Aharonov1984}. It was argued that \enquote{the reduction process \emph{must} be instantaneous} (though this conclusion was challenged by \cite{Mould1999}). \sam{The formalism was further developed by Breuer and Petruccione \cite{breuer_state_1999,breuer_measurements_2007}.} We \sam{continue} in this direction by \sam{writing down} an explicit state update rule in those globally hyperbolic spacetimes which allow for this Tomonaga-Schwinger equation to be written down, and shall analyse different consequences of this rule. 

\begin{definition}[State update rule in the interaction picture]
    \label{def:mmt update rule}
    Let $(\mathcal{M},g)$ be a globally hyperbolic spacetime such that there exists a map $\Sigma \mapsto \hi_\Sigma \cong \hi$ for all acausal Cauchy hypersurfaces $\Sigma$. Suppose the dynamics are given by unitary isomorphisms $U_{\Sigma}^{\Sigma'} : \hi \to \hi$ for every disjoint pair of acausal Cauchy hypersurfaces $\Sigma,\Sigma' \subset \mathcal{M}$. Let $\{\M_k : \Bor(\mathcal{M}) \to \bh\}_{k \in \mathbb{N}}$ be a measurement, and $\U \subset \mathcal{M}$ be precompact such that $B^+(\U) \cap B^-(\U) = \varnothing$. The state update rule in the interaction picture corresponds to the following statement: 
    \begin{quote}
        Let $\Sigma_i$ and $\Sigma_f$ be any disjoint pair of acausal Cauchy hypersurfaces such that $B^-(\U) \subset \Sigma_i$ and $B^+(\U) \subset \Sigma_f$\footnote{These exist by proposition \ref{prop:Bpm is acausal compact} and theorem \ref{Thm: three cauchy hypersurfaces bernal}.}. Let $\tau$ be any Cauchy time function (see theorem \ref{thm:Cauchy time function} in the Appendix) such that $\tau^{-1}(t_i) = \Sigma_i$ and $\tau^{-1}(t_f) = \Sigma_f$, $t_i < t_f$. Upon a measurement $\{\M_k\}$ in the region $\U$ with outcome $j$, the state updates as
    \begin{equation}
        \ket{\psi[\tau^{-1}(t_f')]} = \frac{U^{\tau^{-1}(t_f')}_{\Sigma_f}\M_j(\U) U_{\Sigma_i}^{\Sigma_f} \ket{\psi[\Sigma_i]}}{\norm{\M_j(\U) U_{\Sigma_i}^{\Sigma_f} \ket{\psi[\Sigma_i]}}}
    \end{equation}
    where $t_f' > t_f$, \sam{with probability}
    \begin{equation}
        p_\U(j) := \norm{\M_j(\U) U_{\Sigma_i}^{\Sigma_f} \ket{\psi[\Sigma_i]}}^2 \, .
    \end{equation}
    \end{quote}    
\end{definition}

That is, the state after the measurement is related to the state before through this (selective) state update rule. Importantly, the state on any $\Sigma_f$ which contains $B^+(\U)$ will be the same as the state lying on any $\Sigma_f'$ which satisfies $B^+(\U) \subset \Sigma_f'$; likewise the state on $\Sigma_i$ which contains $B^-(\U)$ ought to be the same as that on any $\Sigma_i'$ which satisfies $B^-(\U) \subset \Sigma_i'$. This can be seen as a natural consequence of the independence of one's foliation of $(\mathcal{M},g)$ on the resulting physics. This, alongside proposition \ref{prop: cov Bpm}, suggests that this state update rule is indeed covariant (especially if we further restrict our discussion to covariant POVMs \cite{Holevo_1982,Chiribella_2004, carette_operational_2023}). Note that we do not attempt to implement the state update rule through physical contents that lie within the theory. Rather, we stay agnostic in a \enquote{Copenhagenist} fashion to the precise physical realisation of such a collapse and to the implementation of the measurement process.  \\

We here mentioned precompact regions rather than compact ones, however there is little to no operational distinction between a region and its closure.\footnote{For example, if quantum gravity (or some similar theory) can be applied to position measurements, then resolving the difference between a region $\{x=(0,1),t=(0,1)\}$ and $\{x=[0,1],t=[0,1]\}$ would require an infinitely precise microscope, which would require an infinite amount of energy and thus create a black hole in that region from energies around the Planck scale and thus destroy the measurement -- the Planck scale being infinitely larger scale than that of a single mathematical point.} One may of course very well work with compact regions in the first place, though if one wants to discuss arbitrary subregions of precompact regions (which are precompact), this is possible whereas arbitrary subregions of compact regions need not be compact (but are precompact). \\

This notion of wavefunction collapse generalises the standard Lüders' rule of nonrelativistic physics, and may straightforwardly be generalised to density operators. Note that this state update rule is insensitive to the form of the unitary state dynamics \emph{over unitarily equivalent Hilbert spaces}: if $\M$ is a measurement, then so is $U_1 \M(\cdot) U_2$ where $U_1$ and $U_2$ are unitaries on $\hi$, as was reminded in proposition \ref{prop: U1 M U2 is a mmt}. This discussion is less straightforward when state dynamics is non-unitary: this is a well-known nontrivial issue already present in Minkowski spacetime \cite{gisin_weinbergs_1990,kent_nonlinearity_2005}, for which the Born rule generally needs to be extended \cite{helou_extensions_2017}. This state update rule also generalises some \enquote{instantaneous} versions of Born rules and collapse rules in globally hyperbolic spacetimes that have been explored e.g. in \cite{lienert_borns_2020,lill_another_2022,reddiger_towards_2025}. However, these generally do not require the Hilbert space at each Cauchy slice to be unitarily equivalent to the Hilbert space at any other Cauchy slice. In such cases, the dynamics is given by unitary isomorphisms $U_\Sigma^{\Sigma'} : \hi_\Sigma \to \hi_{\Sigma'}$ for every pair of disjoint acausal Cauchy\footnote{We here follow \cite{lienert_borns_2020} in imposing acausality whilst allowing for Cauchy hypersurfaces which are not necessarily spacelike. Requiring a timelike normal 1-form everywhere is however perhaps natural for a notion of \enquote{time-evolution}; the arguments in this paper can plausibly be extended to such a further restriction. In particular, if one shows that $B^\pm(\U),B^\pm(\V)$ and/or $\sigma_{1},\sigma_2$ in App. \ref{sec:diff geo} are not only compact and acausal but also spacelike submanifolds with boundary, then one can embed these into spacelike Cauchy hypersurfaces \cite{bernal_further_2006} and the same arguments as in the rest of this paper follow.} hypersurfaces, with
\begin{equation}
    \ket{\psi[\Sigma']} = U_\Sigma^{\Sigma'} \ket{\psi[\Sigma]} \, , \, U^{\Sigma''}_{\Sigma'} U^{\Sigma'}_{\Sigma} = U^{\Sigma''}_\Sigma \, , \, U^\Sigma_\Sigma = \mathbb{1}_{\mathcal{B}(\hi_\Sigma)} \, .
\end{equation}
This form of evolution can be understood, in broad terms, as a Heisenberg-like or Schrödinger-like formulation of the Tomonaga-Schwinger equation which is intrinsically defined in an interaction picture. We shall make the link with unitarily inequivalent evolutions in section \ref{sec: beyond interaction pic}. \\

Although instantaneous measurements are useful to model how experimentalists \enquote{see} collapse occur, we argue that this instantaneity is fundamentally unphysical and unscientific (even if one is not a spacetime atomist). From an operational point of view, measurements always occur in a finite region of space \emph{and} time. In practice, although the effects of measurements cause collapse \enquote{very fast}\footnote{Even this is arguably contentious -- one may model measurement outcomes as determined by the location of an apparatus pointer, ink on paper, or local densities of chemical species in an observer’s brain. Placing bounds on collapse times from experimental investigation is thus, in practice, less obvious than what is sometimes claimed.}, they will never be instantaneous.\footnote{The determination of \emph{exactly when} a collapse-inducing measurement occurs can be, at best, made up to the Planck scale (beyond which a realistic model of position measurements will generate a black hole and destroy the measurement). One can then employ a measure-theoretic argument to argue against the instantaneity of measurements: if one can (potentially) bound the measurement duration time to the time-interval $[0,t_P]$ where $t_P$ is the Planck time, then unless one fine-tunes one's credence for the duration time of measurements to $p(\Delta t=0) = 1$ (which does not and \emph{cannot} have support from data), one ought to consider the (arguably infinitely more likely) possibility that a lower bound on the measurement time is any interval contained in $[0,t_P]$.} Thus, finite-time measurements are, in the author's point of view, a much-better suited notion of measurement for realistic considerations. \\

This discussion, and state update rule, are of course reminiscent of the Hellwig-Kraus reduction rule \cite{Hellwig1970} in the context of local quantum field theory, where it was already discussed that instantaneous measurements are seemingly at odds with (Lorentz) covariance (though instantaneous state updates need not be). There, it is postulated that a field measurement in any \enquote{finite} spacetime region $\U$ changes the field's state in the future and side cones of the region (that is, in $\mathcal{M} \smallsetminus J^-(\U)$). This is certainly analogous to the above statement that relies on acausal Cauchy hypersurfaces which contain $B^\pm(\U)$, with the difference that instead of a slice collapse, Hellwig-Kraus reduction is a light-cone collapse and associates states to regions rather than restricting this to Cauchy hypersurfaces and providing a specific implementation of the collapse or of the dynamics. Here, we have some form of ``side-cone collapse" which happens on every such Cauchy hypersurfaces. Moreover, we do assume some form of non-instantaneity -- this can be covered through a statement of the form: upon a measurement $\{\M_k\}$ in the region $\U$ with outcome $j$,
\begin{multline}
    \label{eqn:instantaneous mmts}
    \lim_{\epsilon \to 0^+} \Bigg(\ket{\psi[\tau^{-1}(t_f + \epsilon)]} \\- \frac{\M_j(\U) U_{\tau^{-1}(t_i-\epsilon)}^{\tau^{-1}(t_f+\epsilon)} \ket{\psi[\tau^{-1}(t_i - \epsilon)]}}{\norm{\M_j(\U) U_{\tau^{-1}(t_i-\epsilon)}^{\tau^{-1}(t_f+\epsilon)} \ket{\psi[\tau^{-1}(t_i - \epsilon)]}}} \Bigg) = 0
\end{multline}
in the weak-* topology of $\bh$ \sam{with Born probability $p_\U(j)$}, provided $B^+(\U) \subset \tau^{-1}(t_f)$ and $B^-(\U) \subset \tau^{-1}(t_i)$, that is, the state of the system \enquote{just after} the measurement is related to that \enquote{just before} the measurement through Lüder's rule (up to some appropriate notion of \enquote{just after} and \enquote{just before} in spacetime). This is also compatible with \enquote{instantaneous} measurements if $t_i = t_f$, and with the notion of state update rule we provided previously. \\

A first objection arises: unlike Hellwig and Kraus, who considered that collapse occurs strictly on the past light cone \cite{hellwig_pure_1969,Hellwig1970} and only assign (algebraic) states to precompact regions rather than to Cauchy hypersurfaces, if a measurement is carried out in a precompact region $\U$, what happens on a hypersurface which intersects the interior $\text{int}(\V)$ of a region $\V$ over which another measurement is conducted? \\

One response to this is to take an operationalist stance: what matters for an experimentalist in $\V$ is the outcome of experiments which is obtained \emph{on (or in the causal future of) the future boundary} of the region on which the measurement is undertaken. If an experimentalist is able to observe some outcome during a measurement, then that measurement can be split into (at least) two individual measurements, since the intermediate observation gives rise to a probability distribution and thus a POVM (and thus a measurement operator). Hence, even for a Cauchy hypersurface which intersects part (but not all of) the (future causal cone of the) future causal boundary of $\V$, the state would remain the post-measurement state after the measurement of $\U$ (up to unitary evolution) but the pre-measurement state with respect to $\V$ on this Cauchy hypersurface unless that measurement can be decomposed. Otherwise, only those Cauchy hypersurfaces which lie in the strict causal future of the future causal boundary of $\V$ would \enquote{see} a state-change from the measurement in $\V$. A similar argument (focusing on the non-covariance of the state histories) was given in \cite{Aharonov1984}, who argued that \enquote{the states themselves make sense only within a given frame or, more abstractly, along some given family of parallel spacelike hypersurfaces, and this is markedly in contrast to the nonrelativistic case}.

\subsection{Spacelike measurements}

Let us now show that, with the state update rule given in definition \ref{def:mmt update rule}, measurement operators indeed anyonically commute over spacelike-separated precompact regions, \sam{consistent with the relativistic structure of the theory}. 

\begin{theorem}
    \label{thm: measurement}
    Let $(\mathcal{M},g)$ be a time-oriented globally hyperbolic $d$-dimensional spacetime such that there exists a map $\Sigma \mapsto \hi_\Sigma \cong \hi$ for all acausal Cauchy hypersurfaces $\Sigma \subset \mathcal{M}$, $\{\M_{1,m} : \Bor(\mathcal{M}) \to \bh\}_{m \in \mathbb{N}}$ and $\{\M_{2,n} : \Bor(\mathcal{M}) \to \bh\}_{n \in \mathbb{N}}$ be two measurements. Then for all precompact $\U, \V \subset \mathcal{M}$ such that $\overline{\U} \indep \overline{\V}$, and \sam{for all respective outcomes $m,n$ of $\{\M_{1,m}\}$ and $\{\M_{2,n}\}$ there exists a $\phi_{mn}(\U,\V) \in [0,2\pi)$ such that}
    \begin{equation}
        \comm{\M_{1,m}(\U)}{\M_{2,n}(\V)}_{\phi_{mn}(\U,\V)} = 0
    \end{equation}
    provided
    \begin{enumerate}
        \item rays in $\hi_\Sigma$ are foliation-independent for all acausal Cauchy hypersurfaces $\Sigma \subset \mathcal{M}$,
        \item the state update rule of definition \ref{def:mmt update rule} holds, and
        \item the dynamics are given by unitary isomorphisms $U_{\Sigma}^{\Sigma'} : \hi \to \hi$ for every pair of disjoint acausal Cauchy hypersurfaces $\Sigma,\Sigma' \subset \mathcal{M}$,
        \item \sam{the dynamics are integrable.}
    \end{enumerate}
\end{theorem}

\begin{proof}
    Let $\U, \V$ be any two precompact regions of $\mathcal{M}$ such that $\overline{\U} \indep \overline{\V}$. By theorem \ref{thm:Foliation Cauchy}, there exists foliations of $\mathcal{M}$ by means of two families $\mathcal{S}_1$ and $\mathcal{S}_2$ of acausal Cauchy hypersurfaces such that 
    \begin{enumerate}
        \item there exists a $\Sigma_1 \in \mathcal{S}_1$ such that $\Sigma_1 \cap J^+(\U) \neq \varnothing$ and $\Sigma_1 \cap J^-(\U) = \varnothing$ while $\Sigma_1 \cap J^+(\V) = \varnothing$ and $\Sigma_1 \cap J^-(\V) \neq \varnothing$, i.e. $\Sigma_1$ is in the strict causal future of $\U$ and the strict causal past of $\V$;
        \item there exists a $\Sigma_2 \in \mathcal{S}_2$ such that $\Sigma_2 \cap J^+(\U) = \varnothing$ and $\Sigma_2 \cap J^-(\U) \neq \varnothing$ and $\Sigma_2 \cap J^+(\V) \neq \varnothing$ and $\Sigma_2 \cap J^-(\V) = \varnothing$, i.e. $\Sigma_2$ is in the strict causal past of $\U$ and the strict causal future of $\V$.
        \item there exists two spacelike Cauchy hypersurfaces $\Sigma_i,\Sigma_f \in \mathcal{S}_1 \cap \mathcal{S}_2$ such that $\Sigma_1,\Sigma_2 \subset I^+(\Sigma_i)$ and $\Sigma_1,\Sigma_2 \subset I^-(\Sigma_f)$, i.e. $\Sigma_i$ and $\Sigma_2$ lie, respectively, in the timelike past and timelike future of both $\Sigma_1$ and $\Sigma_2$.
    \end{enumerate}
    
    \begin{figure*}[t!]
    \centering
    \begin{tikzpicture}[x=0.75pt,y=0.75pt,yscale=-0.9,xscale=0.9]

\draw   (72,144) .. controls (92,134) and (188.67,132.83) .. (258,153.5) .. controls (327.33,174.17) and (217.67,189.67) .. (162,204) .. controls (106.33,218.33) and (121.67,247.67) .. (72,204) .. controls (22.33,160.33) and (52,154) .. (72,144) -- cycle ;
\draw   (470,138.5) .. controls (490,128.5) and (531.67,128) .. (549.67,143.33) .. controls (567.67,158.67) and (648,182.17) .. (659,191.33) .. controls (670,200.5) and (596,214.5) .. (502.67,199.67) .. controls (409.33,184.83) and (450,148.5) .. (470,138.5) -- cycle ;
\draw    (4,69.5) .. controls (44,39.5) and (633,60.5) .. (697,65.5) ;
\draw [color={rgb, 255:red, 74; green, 144; blue, 226 }  ,draw opacity=1 ] [dash pattern={on 4.5pt off 4.5pt}]  (2,101.5) .. controls (84.26,87.72) and (184.91,113.87) .. (291.37,146.55) .. controls (424.34,187.38) and (566.38,238.39) .. (693,234.5) ;
\draw [color={rgb, 255:red, 208; green, 2; blue, 27 }  ,draw opacity=1 ] [dash pattern={on 0.84pt off 2.51pt}]  (4,254.5) .. controls (169,258.5) and (488,82.5) .. (690,88.5) ;
\draw    (20,44) -- (20.91,12.5) ;
\draw [shift={(21,9.5)}, rotate = 91.66] [fill={rgb, 255:red, 0; green, 0; blue, 0 }  ][line width=0.08]  [draw opacity=0] (8.93,-4.29) -- (0,0) -- (8.93,4.29) -- cycle    ;
\draw    (4,283.5) .. controls (44,253.5) and (633,274.5) .. (697,279.5) ;
\draw [color={rgb, 255:red, 74; green, 144; blue, 226 }  ,draw opacity=1 ] [dash pattern={on 4.5pt off 4.5pt}]  (157,262.5) -- (157,252.5) -- (157,217.5) ;
\draw [shift={(157,214.5)}, rotate = 90] [fill={rgb, 255:red, 74; green, 144; blue, 226 }  ,fill opacity=1 ][line width=0.08]  [draw opacity=0] (8.93,-4.29) -- (0,0) -- (8.93,4.29) -- cycle    ;
\draw [color={rgb, 255:red, 74; green, 144; blue, 226 }  ,draw opacity=1 ] [dash pattern={on 4.5pt off 4.5pt}]  (629,228) -- (629,218) -- (629,206.5) ;
\draw [shift={(629,203.5)}, rotate = 90] [fill={rgb, 255:red, 74; green, 144; blue, 226 }  ,fill opacity=1 ][line width=0.08]  [draw opacity=0] (8.93,-4.29) -- (0,0) -- (8.93,4.29) -- cycle    ;
\draw [color={rgb, 255:red, 74; green, 144; blue, 226 }  ,draw opacity=1 ] [dash pattern={on 4.5pt off 4.5pt}]  (619,169) -- (619,159) -- (619,72.5) ;
\draw [shift={(619,69.5)}, rotate = 90] [fill={rgb, 255:red, 74; green, 144; blue, 226 }  ,fill opacity=1 ][line width=0.08]  [draw opacity=0] (8.93,-4.29) -- (0,0) -- (8.93,4.29) -- cycle    ;
\draw [color={rgb, 255:red, 74; green, 144; blue, 226 }  ,draw opacity=1 ] [dash pattern={on 4.5pt off 4.5pt}]  (81,134) -- (81,124) -- (81,107.5) ;
\draw [shift={(81,104.5)}, rotate = 90] [fill={rgb, 255:red, 74; green, 144; blue, 226 }  ,fill opacity=1 ][line width=0.08]  [draw opacity=0] (8.93,-4.29) -- (0,0) -- (8.93,4.29) -- cycle    ;
\draw [color={rgb, 255:red, 208; green, 2; blue, 27 }  ,draw opacity=1 ] [dash pattern={on 0.84pt off 2.51pt}]  (582,268.5) -- (582,262) -- (581.06,212.5) ;
\draw [shift={(581,209.5)}, rotate = 88.91] [fill={rgb, 255:red, 208; green, 2; blue, 27 }  ,fill opacity=1 ][line width=0.08]  [draw opacity=0] (8.93,-4.29) -- (0,0) -- (8.93,4.29) -- cycle    ;
\draw [color={rgb, 255:red, 208; green, 2; blue, 27 }  ,draw opacity=1 ] [dash pattern={on 0.84pt off 2.51pt}]  (599,160.5) -- (599,154) -- (599.94,105.5) ;
\draw [shift={(600,102.5)}, rotate = 91.11] [fill={rgb, 255:red, 208; green, 2; blue, 27 }  ,fill opacity=1 ][line width=0.08]  [draw opacity=0] (8.93,-4.29) -- (0,0) -- (8.93,4.29) -- cycle    ;
\draw [color={rgb, 255:red, 208; green, 2; blue, 27 }  ,draw opacity=1 ] [dash pattern={on 0.84pt off 2.51pt}]  (52,246.5) -- (52,240) -- (51.06,190.5) ;
\draw [shift={(51,187.5)}, rotate = 88.91] [fill={rgb, 255:red, 208; green, 2; blue, 27 }  ,fill opacity=1 ][line width=0.08]  [draw opacity=0] (8.93,-4.29) -- (0,0) -- (8.93,4.29) -- cycle    ;
\draw [color={rgb, 255:red, 208; green, 2; blue, 27 }  ,draw opacity=1 ] [dash pattern={on 0.84pt off 2.51pt}]  (207,132.5) -- (206.04,68.5) ;
\draw [shift={(206,65.5)}, rotate = 89.14] [fill={rgb, 255:red, 208; green, 2; blue, 27 }  ,fill opacity=1 ][line width=0.08]  [draw opacity=0] (8.93,-4.29) -- (0,0) -- (8.93,4.29) -- cycle    ;

\draw (322,30.4) node [anchor=north west][inner sep=0.75pt]    {$\ket{\psi [ \Sigma _{f}]}$};
\draw (137,83.4) node [anchor=north west][inner sep=0.75pt]    {$\Sigma _{1}$};
\draw (514,83.4) node [anchor=north west][inner sep=0.75pt]    {$\Sigma _{2}$};
\draw (136,167.4) node [anchor=north west][inner sep=0.75pt]    {$\{\M_{1,m}(\U)\}_{m \in \mathbb{N}}$};
\draw (504,166.4) node [anchor=north west][inner sep=0.75pt]    {$\{\M_{2,n}(\V)\}_{n \in \mathbb{N}}$};
\draw (30,19) node [anchor=north west][inner sep=0.75pt]   [align=left] {t};
\draw (332,242.4) node [anchor=north west][inner sep=0.75pt]    {$\ket{\psi [ \Sigma _{i}]}$};
\draw (53,106.4) node [anchor=north west][inner sep=0.75pt]    {$U_{2}$};
\draw (167,237.4) node [anchor=north west][inner sep=0.75pt]    {$U_{1}$};
\draw (637,209.4) node [anchor=north west][inner sep=0.75pt]    {$U_{3}$};
\draw (626,120.4) node [anchor=north west][inner sep=0.75pt]    {$U_{4}$};
\draw (554,232.4) node [anchor=north west][inner sep=0.75pt]    {$\tilde{U}_{1}$};
\draw (569,115.4) node [anchor=north west][inner sep=0.75pt]    {$\tilde{U}_{2}$};
\draw (22,201.4) node [anchor=north west][inner sep=0.75pt]    {$\tilde{U}_{3}$};
\draw (216,83.4) node [anchor=north west][inner sep=0.75pt]    {$\tilde{U}_{4}$};
\end{tikzpicture}
	\caption{Foliation of a globally hyperbolic spacetime into two families of Cauchy hypersurfaces $\mathcal{S}_1$ and $\mathcal{S}_2$, both of which contain $\Sigma_i$ and $\Sigma_f$ in the strict causal past and strict causal future, respectively, of spacelike-separated precompact regions $\U$ and $\V$. $\Sigma_1 \in \mathcal{S}_1$ (blue dashed line) lies in the strict causal future of $\U$ and the strict causal past of $\V$; $\Sigma_2 \in \mathcal{S}_2$ (red dotted line) lies in the strict causal past of $\U$ and the strict causal future of $\V$. In the foliation $S_1$, $\ket{\psi[\Sigma_f]}_{12} = c_{12} U_4\M_{2,n}(\V)U_3 U_2 \M_{1,m}(\U) U_1 \ket{\psi[\Sigma_i]}$; in the foliation $S_2$, $\ket{\psi[\Sigma_f]}_{21} = c_{21}\tilde{U}_4\M_{1,m}(\U)\tilde{U}_3 \tilde{U}_2 \M_{2,n}(\V) \tilde{U}_1 \ket{\psi[\Sigma_i]}$. These are one and the same state as the foliation is arbitrary, i.e. they belong to the same ray.}
    \label{fig:povms}
\end{figure*}
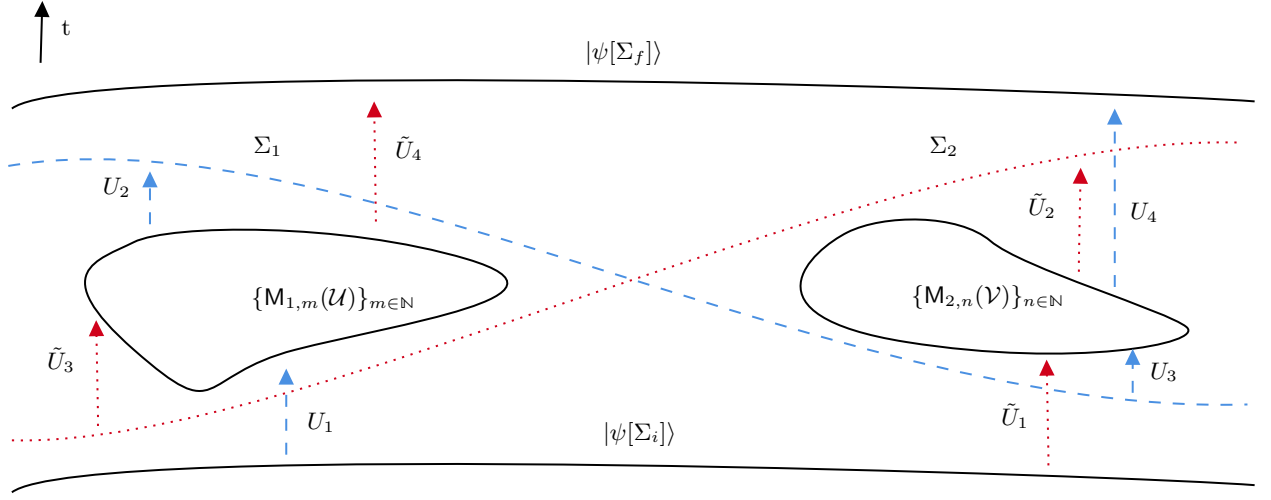
    
    This is shown in Fig. \ref{fig:povms}. Thus, by the state update rule, in the foliation $S_1$, $\M_{1,m}(\U)$ is applied before $\M_{2,n}(\V)$ on any state $\ket{\psi[\Sigma_i]}$ for some outcomes $m,n$ with respective Born probabilities $p_\U(m)$ and $p_\V(n)$, i.e. $\ket{\psi[\Sigma_f]}_{12} = c_{12} U_4\M_{2,n}(\V)U_3 U_2 \M_{1,m}(\U) U_1 \ket{\psi[\Sigma_i]}$ where $c_{12}$ is a normalisation constant and $U_i$ are unitary evolution operators on $\hi$ which depend on the Hamiltonian $H_I$. On the other hand, in the foliation $S_2$, $\M_{2,n}(\V)$ is applied before $\M_{1,m}(\U)$ on any state $\ket{\psi[\Sigma_i]} \in \hi$ with $\ket{\psi[\Sigma_f]}_{21} = c_{21} \tilde{U}_4\M_{1,m}(\U)\tilde{U}_3 \tilde{U}_2 \M_{2,n}(\V) \tilde{U}_1 \ket{\psi[\Sigma_i]}$ where again $c_{21}$ is a normalisation constant and the $\tilde{U}_i$ are unitaries.
    
    Crucially, \sam{once the \emph{unitary integrable} dynamics have been fixed, the individual operators $\M_{1,m}(\U)$ and $\M_{2,n}(\V)$ are definite elements in $\bh$ and can thus be fixed independently of the \sam{\emph{details} of different \emph{unitary integrable}} dynamics that one would like to impose for a given physical situation.} In other words, it suffices to consider the situation where $U_i = \tilde{U}_i = \mathbb{1}$ (take $H_I=0$) when there is no wavefunction dynamics, and deduce the commutation relation between these $\M_{1,m}(\U)$ and $\M_{2,n}(\V)$ which then holds for any choice of \emph{\sam{unitary integrable}} dynamics. We have that
    \begin{align}
        \ket{\psi[\Sigma_f]}_{12} &= \frac{\M_{1,m}(\U) \M_{2,n}(\V) \ket{\psi[\Sigma_i]}}{\norm{\M_{1,m}(\U) \M_{2,n}(\V) \ket{\psi[\Sigma_i]}}} \\ &\equiv c^{\psi[\Sigma_i]}_{\U,\V} \M_{1,m}(\U) \M_{2,n}(\V) \ket{\psi[\Sigma_i]} 
    \end{align}
    where $c^{\psi[\Sigma_i]}_{\U,\V} := \norm{\M_{1,m}(\U) \M_{2,n}(\V) \ket{\psi[\Sigma_i]}}^{-1}$, while
    \begin{align}
        \ket{\psi[\Sigma_f]}_{21} &= \frac{\M_{2,n}(\V) \M_{1,m}(\U) \ket{\psi[\Sigma_i]}}{\norm{\M_{2,n}(\V) \M_{1,m}(\U) \ket{\psi[\Sigma_i]}}} \\ &\equiv c^{\psi[\Sigma_i]}_{\V,\U}  \M_{2,n}(\V) \M_{1,m}(\U)\ket{\psi[\Sigma_i]}  \, .
    \end{align}
    Since the state on $\Sigma_f$ is independent of the foliations of $\mathcal{M}$ which include $\Sigma_f$, each of the $\ket{\psi[\Sigma_f]}_{12}$ and $\ket{\psi[\Sigma_f]}_{21}$ are necessarily operationally equivalent, i.e. they belong to the same ray $[\psi[\Sigma_f]]$. It follows that there exists a $\phi_{mn;\psi[\Sigma_f]}(\U,\V)$ such that
	\begin{align}
	    \frac{\M_{1,m}(\U) \M_{2,n}(\V) \ket{\psi[\Sigma_i]}}{\norm{\M_{1,m}(\U) \M_{2,n}(\V) \ket{\psi[\Sigma_i]}}} &= \ket{\psi[\Sigma_f]}_{12} \\ &= e^{i \phi_{mn;\psi[\Sigma_f]}(\U,\V)} \ket{\psi[\Sigma_f]}_{21} \\ = e^{i \phi_{mn;\psi[\Sigma_f]}(\U,\V)} &\frac{\M_{2,n}(\V) \M_{1,m}(\U) \ket{\psi[\Sigma_i]}}{\norm{\M_{2,n}(\V) \M_{1,m}(\U) \ket{\psi[\Sigma_i]}}} \, .
	\end{align}
    Thus,
    \begin{multline}
        \Bigg(\M_{1,m}(\U) \M_{2,n}(\V) \\- \frac{c^{\psi[\Sigma_i]}_{\V,\U}}{c^{\psi[\Sigma_i]}_{\U,\V}} e^{i \phi_{mn;\psi[\Sigma_f]}(\U,\V)} \M_{2,n}(\V) \M_{1,m}(\U)\Bigg) \ket{\psi[\Sigma_i]} = 0
    \end{multline}
    where we assume that $\ket{\psi[\Sigma_i]} \notin \ker (\M_{1,m}(\U) \M_{2,n}(\V)) \cup \ker (\M_{2,n}(\V) \M_{1,m}(\U))$; otherwise the phase is consistent with being constant. Moreover, let us show that $\frac{c^{\psi[\Sigma_i]}_{\V,\U}}{c^{\psi[\Sigma_i]}_{\U,\V}}$ is constant and deduce that the phase is also state-independent. If $\dim \hi = 1$ we are done ($\bh \cong \mathbb{C}$ is commutative); if $\dim \hi > 1$, let $\ket{\psi[\Sigma_i]} = a \ket{\alpha[\Sigma_i]} + b \ket{\beta[\Sigma_i]}$ where $\ket{\alpha[\Sigma_i]}$ and $\ket{\beta[\Sigma_i]}$ are linearly independent with $a, b \neq 0$ such that $\ket{\alpha[\Sigma_i]},\ket{\beta[\Sigma_i]} \notin (\ker \M_{1,m}(\U)\M_{2,n}(\V)) \cup (\ker \M_{2,n}(\V) \M_{1,m}(\U))$ (otherwise the phase can be taken to be constant and we are done). We have
    \begin{widetext}
        \begin{align}
        \M_{1,m}(\U) \M_{2,n}(\V) \ket{\psi[\Sigma_i]} &= \frac{c^{\psi[\Sigma_i]}_{\V,\U}}{c^{\psi[\Sigma_i]}_{\U,\V}} e^{i \phi_{mn;\psi[\Sigma_f]}(\U,\V)} \M_{2,n}(\V) \M_{1,m}(\U) \ket{\psi[\Sigma_i]} \\
        &= \frac{c^{\psi[\Sigma_i]}_{\V,\U}}{c^{\psi[\Sigma_i]}_{\U,\V}} e^{i \phi_{mn;\psi[\Sigma_f]}(\U,\V)} \Big(\frac{c^{\alpha[\Sigma_i]}_{\U,\V}}{c^{\alpha[\Sigma_i]}_{\V,\U}} e^{-i \phi_{mn;\alpha[\Sigma_f]}(\U,\V)} a \M_{1,m}(\U) \M_{2,n}(\V) \ket{\alpha[\Sigma_i]} \nonumber \\ &\qquad \qquad \qquad \qquad+ \frac{c^{\beta[\Sigma_i]}_{\U,\V}}{c^{\beta[\Sigma_i]}_{\V,\U}} e^{-i \phi_{mn;\beta[\Sigma_f]}(\U,\V)} b \M_{1,m}(\U) \M_{2,n}(\V) \ket{\beta[\Sigma_i]} \Big) \\
        &\stackrel{!}{=}\M_{1,m}(\U) \M_{2,n}(\V) (a \ket{\alpha[\Sigma_i]} + b \ket{\beta[\Sigma_i]}) \\
        \Rightarrow \M_{1,m}(\U) \M_{2,n}(\V) \Big((\frac{c^{\psi[\Sigma_i]}_{\V,\U}c^{\alpha[\Sigma_i]}_{\U,\V}}{c^{\psi[\Sigma_i]}_{\U,\V}c^{\alpha[\Sigma_i]}_{\V,\U}} &e^{i (\phi_{mn;\psi[\Sigma_f]}(\U,\V)-\phi_{\alpha[\Sigma_f]}(\U,\V))} - 1) a \ket{\alpha[\Sigma_i]} \nonumber \\ &+ (\frac{c^{\psi[\Sigma_i]}_{\V,\U}c^{\beta[\Sigma_i]}_{\U,\V}}{c^{\psi[\Sigma_i]}_{\U,\V}c^{\beta[\Sigma_i]}_{\V,\U}} e^{i (\phi_{mn;\psi[\Sigma_f]}(\U,\V)-\phi_{\beta[\Sigma_f]}(\U,\V))} - 1) b \ket{\beta[\Sigma_i]}\Big) = 0 \, .
    \end{align}
    \end{widetext}
    Since $\ket{\alpha[\Sigma_i]}$ and $\ket{\beta[\Sigma_i]}$ are linearly independent and $a,b \neq 0$, this is only possible if both 
    \begin{widetext}
        \begin{equation}
        \begin{cases}
            \frac{c^{\psi[\Sigma_i]}_{\V,\U}c^{\alpha[\Sigma_i]}_{\U,\V}}{c^{\psi[\Sigma_i]}_{\U,\V}c^{\alpha[\Sigma_i]}_{\V,\U}} e^{i (\phi_{mn;\psi[\Sigma_f]}(\U,\V)-\phi_{mn;\alpha[\Sigma_f]}(\U,\V))} = 1 \\
            \frac{c^{\psi[\Sigma_i]}_{\V,\U}c^{\beta[\Sigma_i]}_{\U,\V}}{c^{\psi[\Sigma_i]}_{\U,\V}c^{\beta[\Sigma_i]}_{\V,\U}} e^{i (\phi_{mn;\psi[\Sigma_f]}(\U,\V)-\phi_{mn;\beta[\Sigma_f]}(\U,\V))} = 1
        \end{cases} \Rightarrow \frac{c^{\psi[\Sigma_i]}_{\V,\U}}{c^{\psi[\Sigma_i]}_{\U,\V}} = \frac{c^{\alpha[\Sigma_i]}_{\V,\U}}{c^{\alpha[\Sigma_i]}_{\U,\V}} = \frac{c^{\beta[\Sigma_i]}_{\V,\U}}{c^{\beta[\Sigma_i]}_{\U,\V}} \equiv k(\U,\V) 
    \end{equation}
    \end{widetext}
    for some $k(\U,\V) \in \mathbb{R}^+$ as $\abs{c^{\chi[\Sigma_i]}_{\V,\U}} = c^{\chi[\Sigma_i]}_{\V,\U}$ since these are norms. Hence,
    \begin{multline}
        k(\U,\V) e^{i(\phi_{mn;\psi[\Sigma_f]}(\U,\V) - \phi_{mn;\alpha[\Sigma_f]}(\U,\V))} = 1\\ \Rightarrow k(\U,\V) = e^{-i(\phi_{mn;\psi[\Sigma_f]}(\U,\V) - \phi_{mn;\alpha[\Sigma_f]}(\U,\V))} \\= e^{-i(\phi_{mn;\psi[\Sigma_f]}(\U,\V) - \phi_{mn;\beta[\Sigma_f]}(\U,\V))} \stackrel{!}{\in} \mathbb{R}^+
    \end{multline}
    so $\phi_{mn;\psi[\Sigma_f]}(\U,\V) \equiv  \phi_{mn;\alpha[\Sigma_f]}(\U,\V) \equiv \phi_{mn;\beta[\Sigma_f]}(\U,\V) \mod 2\pi$ and $k(\U,\V) = 1$. Thus, there exists a $\phi_{mn}(\U,\V) \in [0,2\pi)$ such that   
    \begin{equation}
        \comm{\M_{1,m}(\U)}{\M_{2,n}(\V)}_{\phi_{mn}(\U,\V)} \ket{\psi[\Sigma_i]} = 0
    \end{equation}
    for all pure states $\ket{\psi[\Sigma_i]} \in \hi_{\Sigma_i}$, and since operators in $\bh$ are fully determined by their action on pure states and that $\hi_{\Sigma} \cong \hi$ for all acausal Cauchy hypersurfaces $\Sigma$, this concludes the proof.
\end{proof}

Some comments are in order.
\begin{enumerate}
    \item The measurements $\M_1$ and $\M_2$ may in general \emph{arise} from dynamical considerations. For example, if one takes an Everettian point of view to quantum theory, then measurements are described by complicated combinations of unitaries acting on composite systems which are thus indeed dependent on the dynamics. This means that, given an experimental scenario involving such dynamics, one may need to use a measurement $\M$ rather than a measurement $\tilde{\M}$ to implement the physical situation. However, the commutativity properties of the measurement operators themselves, as mathematical entities, are independent of the \sam{details of the unitary integrable dynamics} once they are fixed.
    \item As previously discussed, the existence of the map $\Sigma \mapsto \hi_\Sigma \cong \hi$ for all acausal Cauchy hypersurfaces $\Sigma \subset \mathcal{M}$ of a given foliation $\mathcal{S}$ places a restriction on the type of globally hyperbolic manifold for which this result holds as well as the type of theory that one wishes to work with. 
    \item The independence of the rays $[\psi[\Sigma]]$ on the choice of foliation of $\mathcal{M}$ which includes $\Sigma$ is a natural assumption that relies on the idea that the physics is foliation independent \sam{and requires the integrability of the Tomonaga-Schwinger dynamics for consistency}. Although this principle is rooted in some intuition that nature, being relativistic, does not \enquote{pick out} preferred foliations, this may or may not be strictly required to hold in all formulations of relativistic quantum theory (notably in Bohmian mechanics \cite{Galvan_2015}). One may indeed think of a universe for which $[\psi[\Sigma|\mathcal{S}_1]] \neq [\psi[\Sigma|\mathcal{S}_2]]$. 
    \item In the proof above, we assumed a future directionality of \enquote{time} for convenience. However, the argument is strictly symmetric under \enquote{time} reversal $\Sigma_i \leftrightarrow \Sigma_f$, $I^+ \leftrightarrow I^-$, $D^+ \leftrightarrow D^-$, $J^+ \leftrightarrow J^-$ and $B^+ \leftrightarrow B^-$. This may be relevant to those working with the two-state vector formalism \cite{muga_two-state_2002}.
    \item The physical context in which this commutation relation occurs is certainly necessary. Indeed, it would be nonsensical to talk about \enquote{spacelike commutativity} of maps defined on $\Bor(\mathcal{M})$ which, in itself, does not even possess a causal structure. Likewise, the state update rule of definition \ref{def:mmt update rule} is crucial, simply from an operational point of view: speaking about the commutativity properties of operators without having a notion of measurement is arguably meaningless when these operators are precisely supposed to define measurements. A similar discussion can be made below concerning POVMs.
    \item The assumptions on the regions are natural restrictions to impose to work with operational physics: if the regions are not precompact, then the causal future of these regions may not even meet (e.g. in the $\abs{t}<1$ globally hyperbolic strip of Minkowski spacetime), so the existence of a future region that can compare the order of such past events is not assured -- even so, every finite experiment happens in a (pre)compact region of spacetime, so this is natural to demand. Likewise, the requirements that the closures are spacelike separated is also natural since, as argued previously, there is no operational distinction between a region and its closure. 
\end{enumerate}

We can provide some restrictions on the phase $\phi(\U,\V)$ using the following result.

\begin{theorem}[\cite{Brooke_2002}]
    \label{thm:comm conditions}
    Let $A, B \in \bh$ such that $AB \neq 0$ and $AB = \lambda BA$ for $\lambda \in \mathbb{C}$. Then
    \begin{enumerate}
        \item If $A$ or $B$ is self-adjoint then $\lambda \in \mathbb{R}$;
        \item If both $A$ and $B$ are self-adjoint then $\lambda \in \{-1,1\}$,
        \item If $A$ and $B$ are self-adjoint and one of them is positive, then $\lambda = 1$.
    \end{enumerate}
\end{theorem}

\begin{corollary}
    Under the assumptions of theorem \ref{thm: measurement},
    \begin{enumerate}
        \item if $\M_{1,m}(\U)$ or $\M_{2,n}(\V)$ (or both) is self-adjoint then $\phi_{mn}(\U,\V) \in \{0,\pi\}$,
        \item if $\M_{1,m}(\U)$ and $\M_{2,n}(\V)$ are self-adjoint and one of them is positive, then $\phi_{mn}(\U,\V) = 0$.
    \end{enumerate}
\end{corollary}

For example, $\M_{1,m}(\U)$ and $\M_{2,n}(\V)$ are positive (and thus self-adjoint) if $\M_{1,m}(\U) = \sqrt{\E_{1,m}(\U)}$ and $\M_{2,n}(\V) = \sqrt{\E_{2,n}(\V)}$ where $\E_1$ and $\E_2$ are POVMs, in which case $\comm{\M_{1,m}(\U)}{\M_{2,n}(\V)} = 0$. \\

\sam{Note that the \emph{anyonic} rather than \emph{bosonic} commutativity in general can be surprising (though these are outcome-dependent operator phases at the Kraus level and have little to do with braid statistics), but is actually fully consistent with Einstein causality. Indeed, no-superluminal signalling is ensured as long as the POVMs associated to the measurement operators $\E_k(\U) := \M_k(\U)^\dagger \M_k(\U)$ satisfy a bosonic commutativity over spacelike separations, which is indeed the case despite this anyonic phase, as we now show.}

\subsection{Bosonic commutativity of POVMs} Now that we have shown that, under certain physical assumptions, selective measurements anyonically commute over spacelike-separated regions in a region-dependent fashion, let us highlight that this result propagates to POVMs. 

\begin{theorem}
    \label{thm: Einstein causality POVM}
    Let $(\mathcal{M},g)$ be a time-oriented globally hyperbolic $d$-dimensional spacetime such that there exists a map $\Sigma \mapsto \hi_\Sigma \cong \hi$ for all acausal Cauchy hypersurfaces $\Sigma \subset \mathcal{M}$, $\{\E_{1,m} : \Bor(\mathcal{M}) \to \eh\}_{m \in \mathbb{N}}$ and $\{\E_{2,n} : \Bor(\mathcal{M}) \to \eh\}_{n \in \mathbb{N}}$ be two POVMs. Then for all precompact $\U, \V \subset \mathcal{M}$ such that $\overline{\U} \indep \overline{\V}$ and \sam{for all respective outcomes $m,n$ of the measurements with POVMs $\{\E_{1,m}\}$ and $\{\E_{2,n}\}$,}
    \begin{equation}
        \comm{\E_{1,m}(\U)}{\E_{2,n}(\V)} = 0
    \end{equation}
    provided
    \begin{enumerate}
        \item rays in $\hi_\Sigma$ are foliation-independent for all acausal Cauchy hypersurfaces $\Sigma \subset \mathcal{M}$,
        \item the state update rule of definition \ref{def:mmt update rule} holds, and
        \item the dynamics are given by unitary isomorphisms $U_{\Sigma}^{\Sigma'} : \hi \to \hi$ for every pair of disjoint acausal Cauchy hypersurfaces $\Sigma,\Sigma' \subset \mathcal{M}$,
        \item \sam{the dynamics are integrable}.
    \end{enumerate}
\end{theorem}

\begin{proof}
    Every POVM $\E_i$ has an associated (Kraus) measurement operator $\sqrt{\E_i} = \sqrt{\E_i}^\dagger$. Hence, $\sqrt{\E_{1,m}} = \sqrt{\E_{1,m}}^\dagger$ and $\sqrt{\E_{2,n}} = \sqrt{\E_{2,n}}^\dagger$ are measurements, i.e. by theorem \ref{thm: measurement} we have
    \begin{multline}
        \comm{\sqrt{\E_{1,m}}(\U)}{\sqrt{\E_{2,n}}(\V)}_{\eta_{mn}(\U,\V)} \\ = \comm{\sqrt{\E_{1,m}}(\U)}{\sqrt{\E_{2,n}}(\V)^\dagger}_{\eta_{mn}(\U,\V)} \\ = \comm{\sqrt{\E_{1,m}}(\U)^\dagger}{\sqrt{\E_{2,n}}(\V)}_{\eta_{mn}(\U,\V)}\\  = \comm{\sqrt{\E_{1,m}}(\U)^\dagger}{\sqrt{\E_{2,n}}(\V)^\dagger}_{\eta_{mn}(\U,\V)} = 0
    \end{multline}
    with $\E_{1,m}(\U) := \sqrt{\E_{1,m}}(\U)^\dagger \sqrt{\E_{1,m}}(\U)$ and $\E_{2,n}(\V) := \sqrt{\E_{2,n}}(\V)^\dagger \sqrt{\E_{2,n}}(\V)$. But for all $\ket{\psi} \in \hi$,
    \begin{widetext}
        \begin{align}
        \expval{\E_{1,m}(\U) \E_{2,n}(\V)}{\psi} &= \expval{\sqrt{\E_{1,m}}(\U)^\dagger \sqrt{\E_{1,m}}(\U) \sqrt{\E_{2,n}}(\V)^\dagger \sqrt{\E_{2,n}}(\V)}{\psi} \\
        &= e^{4i\eta_{mn}(\U,\V)} \expval{\sqrt{\E_{2,n}}(\V)^\dagger \sqrt{\E_{2,n}}(\V) \sqrt{\E_{1,m}}(\U)^\dagger \sqrt{\E_{1,m}}(\U) }{\psi} \\
        &= e^{4i\eta_{mn}(\U,\V)}\expval{\E_{2,n}(\V) \E_{1,m}(\U)}{\psi}
    \end{align}
    \end{widetext}
    Defining $\phi_{mn}(\U,\V) := 4 \eta_{mn}(\U,\V) \mod 2\pi$, and since the action of operators are fully determined by their action on pure states, we have
    \begin{equation}
        \comm{\E_{1,m}(\U)}{\E_{2,n}(\V)}_{\phi_{mn}(\U,\V)} = 0 \, .
    \end{equation}
    But since $\E_{1,m}(\U)$ and $\E_{2,n}(\V)$ are positive, this commutation is bosonic by theorem \ref{thm:comm conditions}.
\end{proof}

\section{Quantum no-signalling for non-selective measurements}

\label{sec:non-selective measurements}

\sam{Another important notion to consider in the context of relativistic causality for quantum measurements is that} of non-selective measurements. These involve density operators (possibly mixed states), which in the Tomonaga-Schwinger picture can also be associated to single Cauchy hypersurfaces as $\rho[\Sigma] \in \state$, with dynamics given by $\rho[\Sigma_f] = U_{\Sigma_i}^{\Sigma_f} \rho[\Sigma_i] (U^{\Sigma_f}_{\Sigma_i})^\dagger$. The state update rule for non-selective measurements associated with the state update rule for selective measurements of definition \ref{def:mmt update rule} can be written as follows.

\begin{definition}[State update rule for non-selective measurements in the interaction picture]
    \label{def:mmt update rule non-selective}
    Let $(\mathcal{M},g)$ be a globally hyperbolic spacetime such that there exists a map $\Sigma \mapsto \hi_\Sigma \cong \hi$ for all acausal Cauchy hypersurfaces $\Sigma$. Suppose the dynamics are given by unitary isomorphisms $U_{\Sigma}^{\Sigma'} : \hi \to \hi$ for every disjoint pair of acausal Cauchy hypersurfaces $\Sigma,\Sigma' \subset \mathcal{M}$. Let $\{\M_i : \Bor(\mathcal{M}) \to \bh\}_{i = 1}^{n \in \mathbb{N}}$ be a finite set of measurements, and $\U \subset \mathcal{M}$ be precompact such that $B^+(\U) \cap B^-(\U) = \varnothing$. The state update rule for non-selective measurements in the interaction picture corresponds to the following statement: 
    \begin{quote}
        Let $\Sigma_i$ and $\Sigma_f$ be any disjoint pair of acausal Cauchy hypersurfaces such that $B^-(\U) \subset \Sigma_i$ and $B^+(\U) \subset \Sigma_f$. Let $\tau$ be any Cauchy time function such that $\tau^{-1}(t_i) = \Sigma_i$ and $\tau^{-1}(t_f) = \Sigma_f$. Upon a measurement $\{\M_i\}_{i=1}^n$ in the region $\U$, 
    \begin{multline}
        \rho[\tau^{-1}(t_f')] = \\ U^{\tau^{-1}(t_f')}_{\Sigma_f} \left(\sum_{i=1}^n \M_i(\U) U_{\Sigma_i}^{\Sigma_f} \rho[\Sigma_i] (U_{\Sigma_i}^{\Sigma_f})^\dagger \M_i(\U)^\dagger \right) 
        (U^{\tau^{-1}(t_f')}_{\Sigma_f})^\dagger
    \end{multline}
    where $t_f' > t_f$.
    \end{quote}    
\end{definition}

For simplicity, let us consider a scenario where the measurements are instantaneous. Then, in the spirit of Eqn. \eqref{eqn:instantaneous mmts}, the state update rule for non-selective instantaneous measurements in the interaction picture takes the form: upon a measurement $\{\M_i\}_{i=1}^n$ in the region $\U$,
\begin{multline}
    \label{eqn:instantaneous mmts nonselectve}
    \lim_{\epsilon \to 0^+} \Bigg(\rho[\tau^{-1}(t + \epsilon)] \\ - \sum_{i=1}^n \M_i(\U) U_{\tau^{-1}(t - \epsilon)}^{\tau^{-1}(t + \epsilon)}\rho[\tau^{-1}(t - \epsilon)] (U_{\tau^{-1}(t - \epsilon)}^{\tau^{-1}(t + \epsilon)})^\dagger \M_i(\U)^\dagger \Bigg) = 0
\end{multline}
in the weak-* topology of $\bh$, provided $B^+(\U) = B^{-}(\U) \subset \tau^{-1}(t)$. In the context of non-selective measurements, Einstein causality can be understood as the statement of quantum no-signalling, which takes the following form.

\begin{theorem}
    \label{thm:quantum no-signalling non-selective mmts}
    Let $(\mathcal{M},g)$ be a time-oriented globally hyperbolic $d$-dimensional spacetime such that there exists a map $\Sigma \mapsto \hi_\Sigma \cong \hi$ for all acausal Cauchy hypersurfaces $\Sigma \subset \mathcal{M}$. Let $\U, \V \subset \mathcal{M}$ are two precompact regions with $\overline{\U} \indep \overline{\V}$ and $B^+(\U) = B^-(\U)$ and $B^+(\V) = B^-(\V)$, $\{\M_i : \Bor(\mathcal{M}) \to \bh\}_{i=1}^{n \in \mathbb{N}}$ and $\{\M'_j : \Bor(\mathcal{M}) \to \bh\}_{j=1}^{m \in \mathbb{N}}$ be two sets of measurements, $\E_k'(\V):= \M'_k(\V)^\dagger \M_k'(\V)$ for $k \in \{1, \cdots, m\}$ be POVMs for the latter. Let $\S_1$ be a foliation of $\mathcal{M}$ into acausal Cauchy hypersurfaces parametrised by a Cauchy time function $\tau$. Suppose 
\begin{enumerate}
    \item $\U \cup \V \subset \Sigma \equiv \tau^{-1}(0)$,  
    \item $\sum_{i=1}^n \M_i(\U)^\dagger \M_i(\U) = \mathbb{1}_{\bh}$,
    \item the state update rules for selective measurements (definition \ref{def:mmt update rule}) and for non-selective measurements (definition \ref{def:mmt update rule non-selective}, in particular the instantaneous version Eqn. \eqref{eqn:instantaneous mmts nonselectve}) hold,
    \item the dynamics are given by unitary isomorphisms $U_{\Sigma}^{\Sigma'} : \hi \to \hi$ for every pair of disjoint acausal Cauchy hypersurfaces $\Sigma,\Sigma' \subset \mathcal{M}$,
    \item \sam{the dynamics are integrable,}
    \item states in $\mathscr{D}(\hi_\Sigma)$ are foliation-independent for all acausal Cauchy hypersurfaces $\Sigma \subset \mathcal{M}$.
\end{enumerate}
Then for any $k \in \{1,\cdots,m\}$
    \begin{multline}
        \lim_{\epsilon \to 0} \Bigg(\Tr[\E'_k(\V) \rho_\U[\Sigma_\epsilon]] \\ - \Tr[\E_k'(\V) U_{\Sigma_{-\epsilon}}^{\Sigma_\epsilon} \rho[\Sigma_{-\epsilon}] (U_{\Sigma_{-\epsilon}}^{\Sigma_\epsilon})^\dagger]\Bigg) = 0
    \end{multline}
    where $\Sigma_\epsilon \equiv \tau^{-1}(\epsilon)$ and \begin{equation*}
        \lim_{\epsilon \to 0} \left(\rho_\U[\Sigma_\epsilon] - \sum_{i=1}^n\M_i(\U) U_{\Sigma_{-\epsilon}}^{\Sigma_\epsilon} \rho[\Sigma_{-\epsilon}] (U_{\Sigma_{-\epsilon}}^{\Sigma_\epsilon})^\dagger \M_i(\U)^\dagger \right) = 0
    \end{equation*}
    in the weak-* topology of $\bh$.
\end{theorem}

Of course, the statement is symmetric in $\U \leftrightarrow \V$. 

\begin{proof}
    For all $k \in \{1,\cdots,m\}$, we have
    \begin{widetext}
        \begin{align}
        \lim_{\epsilon \to 0} \Tr[\E'_k(\V) \rho_\U[\Sigma_\epsilon]] \nonumber &= \lim_{\epsilon \to 0}\sum_{i=1}^n \Tr[\M_k'(\V)^\dagger \M_k'(\V) \M_i(\U) U_{\Sigma_{-\epsilon}}^{\Sigma_\epsilon} \rho[\Sigma_{-\epsilon}] (U_{\Sigma_{-\epsilon}}^{\Sigma_\epsilon})^\dagger \M_i(\U)^\dagger] \\
        &= \lim_{\epsilon \to 0}\sum_{i=1}^n \Tr[\M_i(\U)^\dagger \M_k'(\V)^\dagger \M_k'(\V) \M_i(\U) U_{\Sigma_{-\epsilon}}^{\Sigma_\epsilon} \rho[\Sigma_{-\epsilon}] (U_{\Sigma_{-\epsilon}}^{\Sigma_\epsilon})^\dagger] \\
        &\stackrel{\ref{thm: measurement}}{=} \lim_{\epsilon \to 0}\sum_{i=1}^n e^{-i \phi_{ik}(\U,\V)} \Tr[\M_i(\U)^\dagger \M_k'(\V)^\dagger \M_i(\U) \M_k'(\V) U_{\Sigma_{-\epsilon}}^{\Sigma_\epsilon} \rho[\Sigma_{-\epsilon}] (U_{\Sigma_{-\epsilon}}^{\Sigma_\epsilon})^\dagger] \\
        &\stackrel{\ref{thm: measurement}}{=} \lim_{\epsilon \to 0}\sum_{i=1}^n \Tr[\M_k'(\V)^\dagger \M_i(\U)^\dagger \M_i(\U) \M_k'(\V) U_{\Sigma_{-\epsilon}}^{\Sigma_\epsilon} \rho[\Sigma_{-\epsilon}] (U_{\Sigma_{-\epsilon}}^{\Sigma_\epsilon})^\dagger] \\
        &= \lim_{\epsilon \to 0}\Tr\Bigg[\M_k'(\V)^\dagger \underbrace{\Big(\sum_{i=1}^n \M_i(\U)^\dagger \M_i(\U) \Big)}_{=\mathbb{1}_{\bh}} \M_k'(\V) U_{\Sigma_{-\epsilon}}^{\Sigma_\epsilon} \rho[\Sigma_{-\epsilon}] (U_{\Sigma_{-\epsilon}}^{\Sigma_\epsilon})^\dagger]\Bigg] \\
        &= \lim_{\epsilon \to 0}\Tr[\M_k'(\V)^\dagger \M_k'(\V) U_{\Sigma_{-\epsilon}}^{\Sigma_\epsilon} \rho[\Sigma_{-\epsilon}] (U_{\Sigma_{-\epsilon}}^{\Sigma_\epsilon})^\dagger]
    \end{align}
    \end{widetext}
    which concludes the proof.
\end{proof}

That is, the probabilities of measurements at $\V$ are insensitive to the existence of non-selective measurements undertaken over spacelike-separated regions. The analysis for when the measurements are not instantaneous is more involved, and it is an open question to see whether one still recovers quantum no-signalling of this form in these more realistic cases.

\section{Sorkin's impossible measurements}

\label{sec:impossible}

Having discussed the behaviour of measurements undertaken over two spacelike-separated regions, let us now examine the case of problematic measurements undertaken over causally separated regions. Sorkin's impossible measurements \cite{Sorkin1993} is a central topic in the study of measurements in quantum field theory \cite{fewster2024}. Grossly, the argument is as follows: if $\U,\V,\W \subset \mathcal{M}$ are such that $\U \indep \V$ but $\W \in J^+(\U)$ and $\W \in J^-(\V)$, then the outcome statistics obtained by measuring some observable $C$ in $\V$ may then depend on whether or not an observable $A$ in $\U$ was measured, even though $\U$ and $\V$ are spacelike-separated. This seems problematic since, in Sorkin's words, \enquote{By arranging beforehand that $B$ will certainly be measured [in $\W$], someone at [$\U$] could clearly use this dependence of $C$ on $A$ to transmit information “superluminally” to a friend at [$\V$]}. This setup is shown in Fig. \ref{fig:Sorkin}. \\

\begin{figure*}
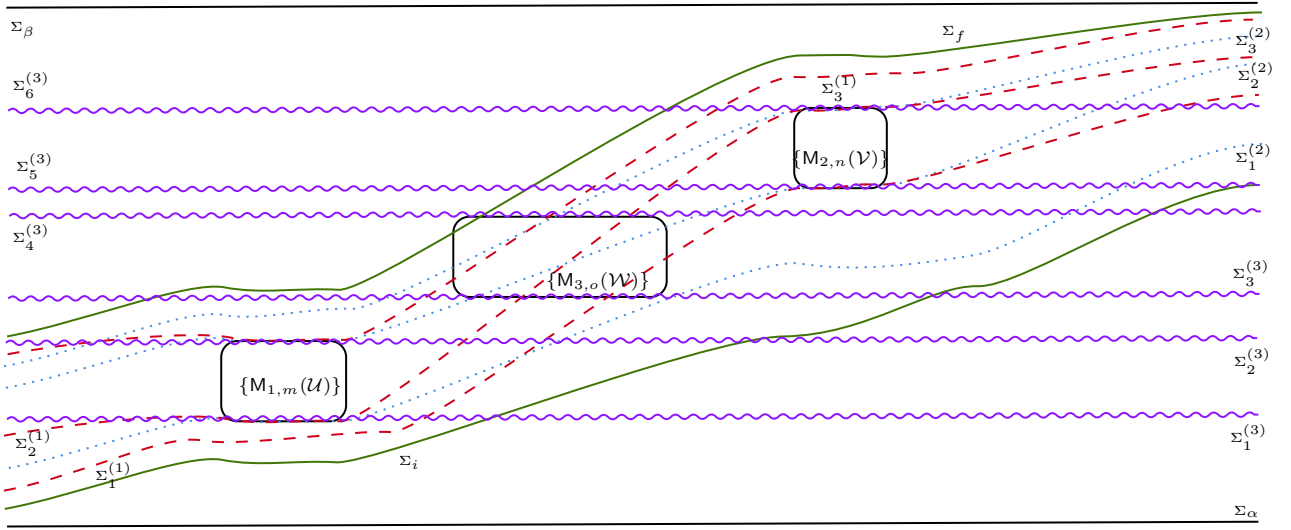

    \centering

    \caption{Setup of Sorkin's impossible measurements. We consider three foliations $\S_1$, $\S_2$ and $\S_3$ of the globally hyperbolic spacetime $(\mathcal{M},g)$, all of which contain the acausal Cauchy hypersurfaces $\Sigma_\alpha$ and $\Sigma_\beta$ (black). $\S_1$ and $\S_2$ both contain $\Sigma_f$ and $\Sigma_i$ (green). Foliation $\S_1$ contains $\Sigma_1^{(1)}$, $\Sigma_2^{(1)}$ and $\Sigma_3^{(1)}$ (red dashed lines). Foliation $\S_2$ contains $\Sigma_1^{(2)}$, $\Sigma_2^{(2)}$ and $\Sigma_3^{(2)}$ (blue dotted lines). Foliation $\S_3$ contains $\Sigma_1^{(3)}$, $\Sigma_2^{(3)}$, $\Sigma_3^{(3)}$, $\Sigma_4^{(3)}$, $\Sigma_5^{(3)}$, $\Sigma_6^{(3)}$ (purple curved lines). In $\S_1$, the state first undergoes the measurement $\{\M_{2,n}(\V)\}$ before the measurement $\{\M_{1,m}(\U)\}$, while in $\S_2$ the state first undergoes the measurement $\{\M_{1,m}(\U)\}$ before the measurement $\{\M_{2,n}(\V)\}$. In $\S_3$, the state first undergoes the measurement $\{\M_{1,m}(\U)\}$, then $\{\M_{3,o}(\W)\}$, then $\{\M_{2,n}(\V)\}$. Note that $\Sigma_f \cap (J^-(B^+(\W)) \smallsetminus B^+(\W)) \neq \varnothing$, i.e. the collapse induced by the measurement of $\{\M_{3,o}(\W)\}$ has not happened on $\Sigma_f$ in either $\S_1$ or $\S_2$.}
    \label{fig:Sorkin}
\end{figure*}

Let us take the statement of the state update rule of definition \ref{def:mmt update rule} seriously. As discussed in section \ref{sec:tomonaga-schwinger mmts}, in a given foliation if an acausal Cauchy hypersurface $\Sigma_f$ has not yet \enquote{fully reached} $B^+(\W)$ (i.e. $\Sigma \cap (J^-(B^+(\W)) \smallsetminus B^+(\W)) \neq \varnothing$), the state has not yet been updated by $\{\M_{3,o}(\W)\}$. Einstein causality of $\{\M_{1,m}(\U)\}$ and $\{\M_{2,n}(\V)\}$ (and their associated POVMs) thus follows as before as long as one can find such a common $\Sigma_f$ in the causal future of both $\U$ and $\V$, contained both in a foliation $\S_1$ in which $\V$ happens before $\U$, and in a foliation $\S_2$ in which $\U$ happens before $\V$. Suppose indeed that both $\U$ and $\V$ are not acausal (in a strong sense where $B^+(\U) \cap B^-(\U) = \varnothing$ and $B^+(\V) \cap B^-(\V) = \varnothing$) \footnote{If parts of one (or both) of these regions is acausal (i.e. there is an instantaneous measurement) then the exact same argument can be made, but the details are more cumbersome to write here.} and are spacelike separated, then \sam{given outcomes $m$ and $n$ for the measurements in $\U$ and $\V$ respectively, it follows that} in $\S_1$ the state on $\Sigma_f$ would be $\ket{\psi[\Sigma_f]} = U_{\Sigma_3^{(1)}}^{\Sigma_f} \M_{1,m}(\U) U_{\Sigma_2^{(1)}}^{\Sigma_3^{(1)}}\M_{2,n}(\V) U_{\Sigma_i}^{\Sigma_2^{(1)}}\ket{\psi[\Sigma_i]}$ while in $\S_2$ it would be $\ket{\psi[\Sigma_f]} = U_{\Sigma_3^{(2)}}^{\Sigma_f} \M_{2,n}(\V) U_{\Sigma_2^{(2)}}^{\Sigma_3^{(2)}} \M_{1,m}(\U) U_{\Sigma_i}^{\Sigma_2^{(2)}}\ket{\psi[\Sigma_i]}$. Turning off the dynamics as before yields Einstein causality of $\{\M_{1,m}(\U)\}$ and $\{\M_{2,n}(\V)\}$. \\

This does not imply that $\M_{1,m}(\U)$ (nor $\M_{2,n}(\V)$) anyonically commutes with any $\M_{3,o}(\W)$ for causally separated $\W$. In fact, \sam{given outcome $o$ in $\W$, it follows that} in a foliation $\S_3$, the state on an acausal Cauchy hypersurface $\Sigma_\beta$, which belongs to all three foliations, in the causal future of all three regions, is $\ket{\psi[\Sigma_\beta]} \propto  U_{\Sigma_6^{(3)}}^{\Sigma_\beta} \M_{2,n}(\V) U_{\Sigma_4^{(3)}}^{\Sigma_5^{(3)}} \M_{3,o}(\W) U_{\Sigma_2^{(3)}}^{\Sigma_3^{(3)}} \M_{1,m}(\U) U_{\Sigma_\alpha}^{\Sigma_2^{(3)}}\ket{\psi[\Sigma_\alpha]}$ where $\Sigma_\alpha$ is an acausal Cauchy hypersurface which belongs to all three foliation in the causal past of all three regions. In all foliations, this state (as a ray) will be the same; turning off the dynamics yields
\begin{multline}
    \M_{2,n}(\V) \M_{3,o}(\W) \M_{1,m}(\U) \propto \M_{3,o}(\W) \M_{2,n}(\V) \M_{1,m}(\U)  \\ \propto \M_{3,o}(\W) \M_{1,m}(\U) \M_{2,n}(\V) \, .
\end{multline}
Two cases now arise.
\begin{enumerate}
    \item $\M_{1,m}(\U)$ is not invertible: in that case, one cannot conclude that $\M_{3,o}(\W) \propto \mathbb{1}_{\bh}$ for any of the operator equalities on $\Sigma_\beta$.
    \item $\M_{1,m}(\U)$ is invertible: in that case, we have that $\comm{\M_{3,o}(\W)}{\M_{2,n}(\V)}_{\phi_{mno}(\U,\V,\W)} = 0$ for some $\phi_{mno}(\U,\V,\W) \in [0,2\pi)$.
\end{enumerate}

The latter case is problematic: one can come up with a setup without $\M_{1,m}(\U)$ for which $\W$ and $\V$ are causally separated, which should \emph{not} generally imply that $\M_{3,o}(\W)$ and $\M_{2,n}(\V)$ anyonically commute. The introduction of another measurement (e.g. which happens to be unitary in $\U$) in the past of $\W$ (but not of $\V$) should not make $\M_{3,o}(\W)$ and $\M_{2,n}(\V)$ suddenly anyonically commute. One way out of this conundrum is to assume that measurements which cause state collapse are \emph{fundamentally} irreversible. For such measurements, Sorkin's paradox is avoided and this state update rule appears to remain consistent. It would be interesting to see whether the case of invertible measurements can also be cured, or whether this highlights something deeper, either about the irreversible nature of quantum measurements or about the limitations of this state update rule.

\section{Einstein causality beyond the interaction picture}

\label{sec: beyond interaction pic}

One step in expanding this framework is to consider measurements beyond the interaction picture, for which the Hilbert space associated to every acausal Cauchy hypersurface need not be unitarily equivalent to that on any other Cauchy slice. 

\begin{definition}[State update rule beyond the interaction picture]
    \label{def:mmt update rule beyond int pic}
    Let $(\mathcal{M},g)$ be a globally hyperbolic spacetime such that there exists a map $\Sigma \mapsto \hi_\Sigma$ for all acausal Cauchy hypersurfaces $\Sigma$. Suppose the dynamics are given by unitary isomorphisms $U_{\Sigma}^{\Sigma'} : \hi_\Sigma \to \hi_{\Sigma'}$ for every disjoint pair of acausal Cauchy hypersurfaces $\Sigma,\Sigma' \subset \mathcal{M}$. Let $\{\M_k : \Bor(\mathcal{M}) \to \bh\}_{k \in \mathbb{N}}$ be a measurement, and $\U \subset \mathcal{M}$ be precompact such that $B^+(\U) \cap B^-(\U) = \varnothing$. The state update rule in the interaction picture corresponds to the following statement: 
    \begin{quote}
        Let $\Sigma_i$ and $\Sigma_f$ be any disjoint pair of acausal Cauchy hypersurfaces such that $B^-(\U) \subset \Sigma_i$ and $B^+(\U) \subset \Sigma_f$. Let $\tau$ be any Cauchy time function such that $\tau^{-1}(t_i) = \Sigma_i$ and $\tau^{-1}(t_f) = \Sigma_f$. Upon a measurement $\{\M_k\}$ in the region $\U$ \sam{with outcome $j$,} the state updates as 
    \begin{equation}
        \ket{\psi[\tau^{-1}(t_f')]} = \frac{U^{\tau^{-1}(t_f')}_{\Sigma_f}\M_j(\U) U_{\Sigma_i}^{\Sigma_f} \ket{\psi[\Sigma_i]}}{\norm{\M_j(\U) U_{\Sigma_i}^{\Sigma_f} \ket{\psi[\Sigma_i]}}}
    \end{equation}
    where $t_f' > t_f$, \sam{with probability}
    \begin{equation}
        p_\U(j) := \norm{\M_j(\U) U_{\Sigma_i}^{\Sigma_f} \ket{\psi[\Sigma_i]}}^2 \, .
    \end{equation}
    \end{quote}    
\end{definition}

The above again links to some of the instantaneous measurement rules introduced in the literature (e.g. \cite{lienert_borns_2020,lill_another_2022,reddiger_towards_2025}), for which the \enquote{curved Born rule} applies for ideal detectors along arbitrary curved Cauchy hypersurfaces. The exact same argument as in section \ref{sec:Eins caus} can thus be made in this context. \\

\begin{theorem}
    Let $(\mathcal{M},g)$ be a time-oriented globally hyperbolic $d$-dimensional spacetime such that there exists a map $\Sigma \mapsto \hi_\Sigma$ for all acausal Cauchy hypersurfaces $\Sigma \subset \mathcal{M}$. Let
\begin{enumerate}
    \item $\{\M_{1,m}^{\Sigma} : \Bor(\mathcal{M}) \to \mathcal{B}(\hi_\Sigma)\}_{m \in \mathbb{N}}$ and $\{\M^{\Sigma}_{2,n} : \Bor(\mathcal{M}) \to \mathcal{B}(\hi_\Sigma)\}_{n \in \mathbb{N}}$ be two measurements,
    \item the dynamics are given by unitary isomorphisms $U_{\Sigma}^{\Sigma'} : \hi_\Sigma \to \hi_{\Sigma'}$ for every pair of disjoint acausal Cauchy hypersurfaces $\Sigma,\Sigma' \subset \mathcal{M}$,
    \item $\hi_{\Sigma_2^{21}} \cong \hi_{\Sigma_3^{12}}$ and $\hi_{\Sigma_3^{21} }\cong \hi_{\Sigma_{2}^{12}}$\footnote{This assumption is physically relevant if one considers that for an acausal compact region $\U$ contained in any two Cauchy hypersurfaces $\Sigma_1$ and $\Sigma_2$, the measurement operators $\M^{\Sigma_1}(\U)$ and $\M^{\Sigma_2}(\U)$ in the same region $\U$ have the same codomain.},
    \item $\U, \V \subset \mathcal{M}$ are precompact such that $\overline{\U} \indep \overline{\V}$. For conciseness, take $B^+(\U) \cap B^-(\U) = \varnothing$ and $B^+(\V) \cap B^-(\V) = \varnothing$,
    \item $\Sigma_i, \Sigma_1^{21},\Sigma_1^{12},\Sigma_2^{21},\Sigma_2^{12},\Sigma_3^{12},\Sigma_3^{21},\Sigma_f$ be disjoint acausal Cauchy hypersurfaces such that
    \begin{itemize}
        \item $\Sigma_i \subset I^-(\Sigma_1^{12}) \cap I^-(\Sigma_1^{21})$, $\Sigma_f \subset I^+(\Sigma_3^{12}) \cap I^+(\Sigma_3^{21})$,
        \item $\Sigma_2^{12} \subset I^+(\Sigma_1^{12})$, $\Sigma_2^{21} \subset I^+(\Sigma_1^{21})$, $\Sigma_3^{12} \subset I^+(\Sigma_2^{12})$, $\Sigma_3^{21} \subset I^+(\Sigma_2^{21})$,
        \item $B^-(\U) \subset \Sigma_1^{12}$ and $B^-(\U) \subset \Sigma_2^{21}$ and $B^+(\U) \subset \Sigma_2^{12}$,
        \item $B^-(\V) \subset \Sigma_1^{21}$ and $B^-(\V) \subset \Sigma_2^{12}$ and $B^+(\V) \subset \Sigma_2^{21}$ and $B^+(\V) \subset \Sigma_3^{12}$.
    \end{itemize}
\end{enumerate}
If rays in $\hi_\Sigma$ are foliation-independent for all acausal Cauchy hypersurfaces $\Sigma \subset \mathcal{M}$, \sam{the dynamics are integrable} and the state update rule of definition \ref{def:mmt update rule beyond int pic} holds, then for all outcomes $m, n$ of the respective $\{\M^{\Sigma}_{1,m}\}$ and $\{\M^\Sigma_{2,n}\}$ measurements,
\begin{multline}
    U^{\Sigma_f}_{\Sigma_3^{21}} \M_{1,m}^{\Sigma_3^{21}}(\U) U_{\Sigma_{2}^{21}}^{\Sigma_{3}^{21}} \M^{\Sigma_2^{21}}_{2,n}(\V) U^{\Sigma_2^{21}}_{\Sigma_i}\\ = e^{i\phi_{mn}(U,\Sigma_i,\Sigma_f;\U,\V)} U^{\Sigma_f}_{\Sigma_3^{12}}  \M^{\Sigma_3^{12}}_{2,n}(\V) U_{\Sigma_{2}^{12}}^{\Sigma_{3}^{12}} \M^{\Sigma_2^{12}}_{1,m}(\U) U^{\Sigma_2^{12}}_{\Sigma_i}
\end{multline}
for some $\phi_{mn}(U,\Sigma_i,\Sigma_f;\U,\V) \in [0,2\pi)$.
\end{theorem}

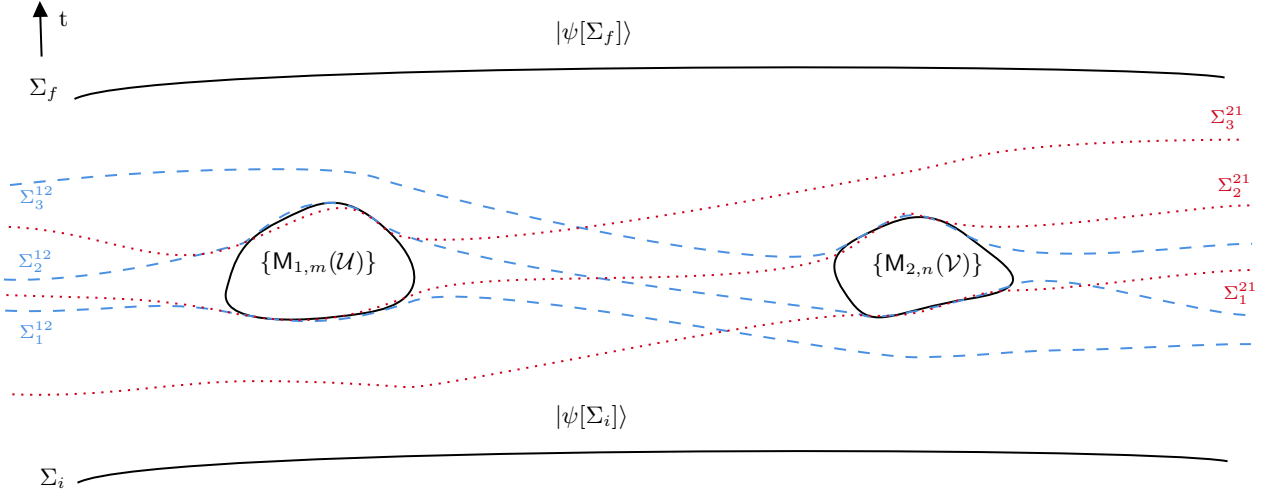
\begin{figure*}[t!]
    \centering
    \begin{tikzpicture}[x=0.75pt,y=0.75pt,yscale=-0.9,xscale=0.9] 
\draw    (39.67,67) .. controls (82.33,45.67) and (633,46.33) .. (681,55) ;
\draw    (41.33,281) .. controls (84,259.67) and (634.67,260.33) .. (682.67,269) ;
\draw   (153,135.67) .. controls (173,125.67) and (185.67,117.67) .. (208,135) .. controls (230.33,152.33) and (244,179.67) .. (201.67,186.33) .. controls (159.33,193) and (119.67,191.67) .. (124.33,173.67) .. controls (129,155.67) and (133,145.67) .. (153,135.67) -- cycle ;
\draw   (478.33,143.67) .. controls (498.33,133.67) and (514.67,125.67) .. (537,143) .. controls (559.33,160.33) and (581.33,171.67) .. (539,178.33) .. controls (496.67,185) and (485.67,197) .. (473,180.33) .. controls (460.33,163.67) and (458.33,153.67) .. (478.33,143.67) -- cycle ;
\draw    (21.33,43.33) -- (20.43,15.33) ;
\draw [shift={(20.33,12.33)}, rotate = 88.15] [fill={rgb, 255:red, 0; green, 0; blue, 0 }  ][line width=0.08]  [draw opacity=0] (8.93,-4.29) -- (0,0) -- (8.93,4.29) -- cycle    ;
\draw [color={rgb, 255:red, 74; green, 144; blue, 226 }  ,draw opacity=1 ] [dash pattern={on 4.5pt off 4.5pt}]  (1,185) .. controls (55.67,187) and (83,177.67) .. (124.33,185.67) .. controls (165.67,193.67) and (190.33,193) .. (228.33,179.67) .. controls (266.33,166.33) and (469,211.67) .. (509.67,211) .. controls (550.33,210.33) and (555,207.67) .. (593.67,207) .. controls (632.33,206.33) and (662.33,203.67) .. (695.67,203.67) ;
\draw [color={rgb, 255:red, 74; green, 144; blue, 226 }  ,draw opacity=1 ] [dash pattern={on 4.5pt off 4.5pt}]  (0.33,167.67) .. controls (55,169.67) and (126.33,149.67) .. (133,147.67) .. controls (139.67,145.67) and (172.67,107) .. (208,135) .. controls (243.33,163) and (472.33,183.67) .. (479.67,187.67) .. controls (487,191.67) and (537.67,179.67) .. (564.33,170.33) .. controls (591,161) and (664.33,187.67) .. (697.67,187.67) ;
\draw [color={rgb, 255:red, 74; green, 144; blue, 226 }  ,draw opacity=1 ] [dash pattern={on 4.5pt off 4.5pt}]  (3.67,115) .. controls (57,109) and (93,107) .. (132,106.33) .. controls (171,105.67) and (196.67,106.33) .. (210.33,111.67) .. controls (224,117) and (429,169.67) .. (465,151) .. controls (501,132.33) and (501.67,123) .. (537,143) .. controls (572.33,163) and (663,147.67) .. (696.33,147.67) ;
\draw [color={rgb, 255:red, 208; green, 2; blue, 27 }  ,draw opacity=1 ] [dash pattern={on 0.84pt off 2.51pt}]  (5,231.67) .. controls (59.67,233.67) and (95.67,224.33) .. (139,224.33) .. controls (182.33,224.33) and (204.33,226.33) .. (224.67,227.67) .. controls (245,229) and (435.67,185) .. (479.67,187.67) .. controls (523.67,190.33) and (526.33,174.33) .. (565,173.67) .. controls (603.67,173) and (665,162.33) .. (698.33,162.33) ;
\draw [color={rgb, 255:red, 208; green, 2; blue, 27 }  ,draw opacity=1 ] [dash pattern={on 0.84pt off 2.51pt}]  (1.67,176.33) .. controls (56.33,178.33) and (81.67,177) .. (124.33,185.67) .. controls (167,194.33) and (190.33,191) .. (231.67,175) .. controls (273,159) and (432.33,177) .. (472.67,149.67) .. controls (513,122.33) and (494.33,131) .. (533.67,137) .. controls (573,143) and (661.67,126.33) .. (695,126.33) ;
\draw [color={rgb, 255:red, 208; green, 2; blue, 27 }  ,draw opacity=1 ] [dash pattern={on 0.84pt off 2.51pt}]  (3,138.33) .. controls (57.67,140.33) and (89.67,165.67) .. (133,147.67) .. controls (176.33,129.67) and (189.67,117.67) .. (213,139.67) .. controls (236.33,161.67) and (463.67,114.33) .. (481.33,111.67) .. controls (499,109) and (509,107) .. (542.33,99) .. controls (575.67,91) and (657,89.67) .. (690.33,89.67) ;

\draw (306,236.07) node [anchor=north west][inner sep=0.75pt]    {$\ket{\psi [ \Sigma _{i}]}$};
\draw (18.67,271.4) node [anchor=north west][inner sep=0.75pt]    {$\Sigma _{i}$};
\draw (12.67,54.07) node [anchor=north west][inner sep=0.75pt]    {$\Sigma _{f}$};
\draw (29.33,16.67) node [anchor=north west][inner sep=0.75pt]   [align=left] {t};
\draw (306,22.07) node [anchor=north west][inner sep=0.75pt]    {$\ket{\psi [ \Sigma _{f}]}$};
\draw (141.33,150.07) node [anchor=north west][inner sep=0.75pt]    {$\{\M_{1,m}(\U)\}$};
\draw (482.67,151.4) node [anchor=north west][inner sep=0.75pt]    {$\{\M_{2,n}(\V)\}$};
\draw (6.67,149.73) node [anchor=north west][inner sep=0.75pt]  [font=\scriptsize,color={rgb, 255:red, 74; green, 144; blue, 226 }  ,opacity=1 ]  {$\Sigma _{2}^{12}$};
\draw (7.33,189.73) node [anchor=north west][inner sep=0.75pt]  [font=\scriptsize,color={rgb, 255:red, 74; green, 144; blue, 226 }  ,opacity=1 ]  {$\Sigma _{1}^{12}$};
\draw (7.33,114.4) node [anchor=north west][inner sep=0.75pt]  [font=\scriptsize,color={rgb, 255:red, 74; green, 144; blue, 226 }  ,opacity=1 ]  {$\Sigma _{3}^{12}$};
\draw (678.67,167.73) node [anchor=north west][inner sep=0.75pt]  [font=\scriptsize,color={rgb, 255:red, 208; green, 2; blue, 27 }  ,opacity=1 ]  {$\Sigma _{1}^{21}$};
\draw (675.33,107.73) node [anchor=north west][inner sep=0.75pt]  [font=\scriptsize,color={rgb, 255:red, 208; green, 2; blue, 27 }  ,opacity=1 ]  {$\Sigma _{2}^{21}$};
\draw (671.33,69.73) node [anchor=north west][inner sep=0.75pt]  [font=\scriptsize,color={rgb, 255:red, 208; green, 2; blue, 27 }  ,opacity=1 ]  {$\Sigma _{3}^{21}$};

\end{tikzpicture}

    \caption{Foliation of a globally hyperbolic spacetime into two families of Cauchy hypersurfaces $\mathcal{S}_1$ and $\mathcal{S}_2$, both of which contain $\Sigma_i$ and $\Sigma_f$ in the strict causal past and strict causal future, respectively, of spacelike-separated precompact regions $\U$ and $\V$. In blue dashed lines ($\mathcal{S}_1$), $B^-(\U) \subset \Sigma_1^{12}$ lies in the causal past of $\V$, $B^+(\U) \cup B^-(\V) \subset \Sigma_2^{12}$ and $B^+(\V) \subset \Sigma_3^{12}$ lies in the causal future of $\U$; in red dotted lines ($\mathcal{S}_2$), $B^-(\V) \subset \Sigma_1^{21}$ lies in the causal past of $\U$, $B^+(\V) \cup B^-(\U) \subset \Sigma_2^{21}$ and $B^+(\U) \subset \Sigma_3^{21}$ lies in the causal future of $\V$. The state on $\Sigma_f$ is independent of the choice of foliation, i.e. the final states in both foliations belong to the same ray.}
    \label{fig:beyond interaction pic}
\end{figure*}

The setup is shown in Fig. \ref{fig:beyond interaction pic}. This dynamical form of an \enquote{anyonic commutation relation of measurements} is perhaps less aesthetic than that which we \sam{recovered} in theorem \ref{thm: measurement}, but is valid beyond the interaction picture \sam{and ensures the consistency of the relativistic physics}. It is explicitly dependent on the dynamics $U$ (perhaps generated by some Hamiltonian $H$, which will change the resulting phase). One could argue that in the free theory, this picture and the interaction pictures coincide \cite{lienert_borns_2020}, so that the commutation relation in the interacting case in such a Schrödinger-like picture may remain identical; however, it is now more tricky to \enquote{shut down} the dynamics as we did previously, since we now \enquote{jump} from Hilbert space to Hilbert space (so we cannot write $U_{\Sigma_1}^{\Sigma_2} = \mathbb{1}$), and $\M_1$ and $\M_2$ now generally map to different Hilbert spaces. Further note that the dependence on the initial and final hypersurfaces is, in a way, \enquote{superficial} in that given other initial and final hypersurfaces $\Sigma_i'$ and $\Sigma_f'$, respectively, we have $\phi_{mn}(U,\Sigma_i,\Sigma_f;\U,\V) \equiv \eta_{mn}(U,\Sigma_i,\Sigma_i',\Sigma_f,\Sigma_f') + \phi_{mn}(U,\Sigma_i',\Sigma_f';\U,\V) \mod 2\pi$ for some $\eta_{mn}(U,\Sigma_i,\Sigma_i',\Sigma_f,\Sigma_f') \in [0,2\pi)$ by unitary evolution. \\

In general, if one does not assume the unitary equivalence of Hilbert spaces across disjoint regions of spacetime, then the strict notion of \enquote{commutativity of observables across spacelike separations} should be weakened to something like the above. Of course, in such situations, the formalism of AQFT may simply be better suited \cite{Torre_1999}, and measurement schemes expressed in that language \cite{fewster_quantum_2020,ruep_causality_2022,fewster_measurement_2023} may allow a more general treatment of the notion of Einstein causality in generic gauge theories and globally hyperbolic spacetimes. \sam{Furthermore, the machinery of relativistic quantum instruments in spacetime \cite{ozawa_quantum_1984,okamura_measurement_2015} may help further understand such general scenarios, and it would be interesting to reframe the discussion of this paper within quantum instrument theory.}

\section*{Acknowledgements}

The author thanks Adrian Kent, Subhayan Roy Moulik and Nicola Pranzini for useful comments and discussions, and Miguel Sánchez Caja for his insights into globally hyperbolic spacetimes and Cauchy hypersurfaces. The author also thanks Masanao Ozawa for his helpful feedback on non-selective quantum measurements in spacetime. The author is funded by a studentship from the Engineering and Physical Sciences Research Council.

\begin{appendix}
    \section{Appendix: Some tools from Lorentzian geometry}

\label{sec:diff geo}

This appendix serves as a reminder of tools from Lorentzian geometry.

\begin{definition}
	A spacetime $(\mathcal{M},g)$ is said to be \emph{globally hyperbolic} if it contains no closed causal curve (i.e. it is causal) and if $\forall p,q \in \mathcal{M}$, the double cone (causal diamond) $J^-(p) \cap J^+(q)$ is compact \cite{Hounnonkpe_2019}.
\end{definition}

Importantly, $(\mathcal{M},g)$ is globally hyperbolic if and only if there exists a smooth spacelike Cauchy hypersurface\footnote{A Cauchy hypersurface $\Sigma$ is a subset of $\mathcal{M}$ such that $D(\Sigma) = \mathcal{M}$, that is, data on $\Sigma$ fully determines the future and past data upon evolution in $\mathcal{M}$.} $\Sigma$ in $\mathcal{M}$ \cite{geroch_domain_1970,bernal_smooth_2003}. Equivalently, any globally hyperbolic manifold $\mathcal{M}$ is diffeomorphic to $\mathbb{R} \times \Sigma$. The existence of a Cauchy temporal function\footnote{A ($C^k$-)Cauchy temporal function is a ($C^k$-)smooth map $\tau : \mathcal{M} \to \mathbb{R}$ with timelike gradient everywhere with spacelike Cauchy hypersurfaces as levels.} is guaranteed in globally hyperbolic spacetimes \cite{bernal_smoothness_2005}. More generally, if $\Sigma$ is an acausal\footnote{$H \subset \mathcal{M}$ is \emph{acausal} if no two points of $H$ are connected by a causal curve.} Cauchy hypersurface in $\mathcal{M}$ (but is not necessarily spacelike\footnote{A hypersurface $H$ is \emph{spacelike} if its normal $1$-form $n_a$ is everywhere timelike.}), then the existence of a Cauchy time function (defined below) which contains $\Sigma$ as a level is ensured.

\begin{theorem}[\cite{bernal_further_2006}]
    \label{thm:Cauchy time function}
    Let $(\mathcal{M},g)$ be globally hyperbolic, and $\Sigma$ an acausal Cauchy hypersurface. Then there exists a smooth function $\tau : \mathcal{M} \to \mathbb{R}$, called a \emph{Cauchy time function}, such that the levels 
    $\Sigma_t = \tau^{-1}(t), t \in \mathbb{R}$, satisfy
    \begin{enumerate}
        \item $\Sigma = \Sigma_0$,
        \item Each $\Sigma_t$ is a smooth spacelike Cauchy hypersurface for any $t \in \mathbb{R^*}$,
        \item $\tau$ is strictly increasing on any future-directed causal curve.
    \end{enumerate}
\end{theorem}

\begin{lemma}
    \label{lem:Cauchy in the future}
    If $\U \indep \V \subset \mathcal{M}$ are precompact, then there exists a spacelike Cauchy hypersurface in the strict causal future of both regions, and a spacelike Cauchy hypersurface in the strict causal past of both regions.
\end{lemma}

\begin{proof}
    It is enough to consider any Cauchy temporal function $\tau$. Indeed, the open subsets $T_m:= \{\tau^{-1}(t) \mid t \in(-m,m), m \in \mathbb{N}\}$ yield an open covering of $\overline{\U}, \overline{\V}$. Therefore, there exists an $m_0 \in \mathbb{N}$ such that both $\overline{\U}$ and $\overline{\V}$ lie in $T_{m_0}$. Thus, any slice $\Sigma_\tau$ with $\tau=c$, where $c > m_0$ is a constant, is Cauchy and lies in the strict causal future of both $\U$ and $\V$. The exact same argument can be held for showing the existence of a spacelike Cauchy hypersurface in the causal past of both regions.
\end{proof} 

Note that if $\U$ or $\V$ are not precompact then this does not necessarily hold. For example, consider the strip $\abs{t} < 1$ in Minkowski spacetime -- this strip is globally hyperbolic. However, if one chooses two spacelike-separated points $p$ and $q$ of this strip which lie sufficiently far away in the $x$ direction, and $\U = J^+(p)$ and $\V = J^+(q)$ then no point of the strip lies in the (strict) causal future of $\U$ or $\V$.\footnote{The author thanks Miguel Sánchez Caja for providing the above proof and this counterexample in a private communication.} We now review a key result which states that any collection of disjoint $C^k$ spacelike Cauchy hypersurfaces can be \enquote{joined} in a single foliation of spacetime through a single Cauchy temporal function.

\begin{theorem}[\cite{muller_note_2016,Muller2012}]
    \label{thm:muller}
    Let, for $n \in \mathbb{N} \cup \{\infty\}$, $(\mathcal{M}, g)$ be a $C^n$ globally hyperbolic spacetime. Let $k \leq n$, $m \in \mathbb{N}$ and let $(\Sigma_{-} = \Sigma_0, \Sigma_1,..., \Sigma_m, \Sigma_{+} = \Sigma_{m+1})$ be an $(m +2)$-tuple of $C^k$ spacelike Cauchy hypersurfaces with $\Sigma_{i+1} \subset I^+(\Sigma_i)$ for all $0 \leq i \leq m +1$. Let $a = (a_1,... a_m)$ be an $m$-tuple of real numbers with $a_{i+1} > a_i$. Then there is a $C^{k-1}$ Cauchy temporal function $T$ with $\Sigma_i = T^{-1}(a_i)$ for all $i \in \{1,..., m\}$.
\end{theorem}

Below, we shall take globally hyperbolic spacetimes to be $C^\infty$ for simplicity. One can also show that any acausal compact region of a globally hyperbolic spacetime can be extended to an acausal Cauchy hypersurface \cite{bernal_further_2006}. This, combined with the above and with the fact that the region between the two spacelike Cauchy hypersurfaces is a globally hyperbolic spacetime, gives rise to the following result.

\begin{theorem}[\cite{bernal_further_2006}]
    \label{Thm: three cauchy hypersurfaces bernal}
    Let $K$ be a compact acausal region in $(\mathcal{M},g)$, $\Sigma_i$, $\Sigma_f$ be two disjoint spacelike Cauchy hypersurfaces such that $K \subset J^+(\Sigma_i) \cap J^-(\Sigma_f)$. Then there exists a Cauchy time function $\tau : \mathcal{M} \to \mathbb{R}$ such that $\tau^{-1}(t_i) = \Sigma_i$, $\tau^{-1}(t_K) = \Sigma_K \supset K$ and $\tau^{-1}(t_f) = \Sigma_f$ where $t_i < t_K < t_f \in \mathbb{R}$.
\end{theorem}

Keeping these results in mind for later, let us highlight that we can indeed embed the future and past boundaries of (pre)compact regions within acausal Cauchy hypersurfaces.

\begin{proposition}
    \label{prop:Bpm is acausal compact}
    Let $\U$ be a precompact region of a globally hyperbolic spacetime $(\mathcal{M},g)$. Then $B^\pm(\U)$ are compact acausal regions in $(\mathcal{M},g)$.
\end{proposition}

\begin{proof}
    We need to show that $B^+(\U)$ is an acausal compact region -- the proof for $B^-(\U)$ is completely analogous.

    \begin{enumerate}
        \item Acausality: Suppose for contradiction that there exist two distinct points $p, q \in B^+(\U)$ and a causal curve $\gamma$ from $p$ to $q$. Then $q \in J^+(p)$. Since $q \in B^+(\U) \subseteq \overline{\U}$, we have $q \in J^+(p) \cap \overline{\U}$. But for $p \in B^+(\U)$, by definition $J^+(p) \cap \overline{\U} = \{p\}$. Therefore, $q=p$, which contradicts our assumption that $p$ and $q$ are distinct. Hence, there is no causal curve between any two distinct points in $B^+(\U)$. Thus, $B^+(\U)$ is acausal.
        
        \item Compactness: 
        
        Since $(\mathcal{M},g)$ is globally hyperbolic, there exists a smooth global time function $\tau: \M \to \R$ whose gradient $\nabla \tau$ is a past-directed timelike vector field. $\tau$ is strictly increasing along any future-directed causal curve from a point $p$ to a distinct point $q$, i.e. $\tau(q) > \tau(p)$.

    The set $\overline{\U}$ is compact, so the continuous function $\tau$ must attain a maximum value on it. Let $\tau_{\max} = \sup_{p \in \overline{\U}} \tau(p)$. We claim that if $p \in B^+(\U)$, then it is necessary that $\tau(p) = \tau_{\max}$.
    
    Indeed, suppose, for contradiction, that there is a point $p \in B^+(\U)$ with $\tau(p) < \tau_{\max}$. Since the maximum of $\tau$ on $\overline{\U}$ is strictly greater than $\tau(p)$, there exists a future-directed timelike curve $\gamma$ starting at $p$ which must remain in $\overline{\U}$ for some parameter length $\epsilon > 0$. Let $q = \gamma(\epsilon)$. Then $q \in I^+(p)$, $q \in \overline{\U}$, and $q \neq p$. This means $J^+(p) \cap \overline{\U}$ contains more than just $\{p\}$, which contradicts the assumption that $p \in B^+(\U)$. Thus, the claim holds.

    We now show that $B^+(\U)$ is closed by showing it contains all its limit points. Let $\{p_n\}_{n \in \mathbb{N}}$ be a sequence in $B^+(\U)$ that converges to a limit $p \in \overline{\U}$. We must show that $p \in B^+(\U)$.
    From the previous claim, since every $p_n \in B^+(\U)$, we must have $\tau(p_n) = \tau_{\max}$ for all $n$. By the continuity of $\tau$, we have:
    \begin{equation}
        \tau(p) = \tau(\lim_{n \to \infty} p_n) = \lim_{n \to \infty} \tau(p_n) = \tau_{\max} \, .
    \end{equation}
    Now, assume for contradiction that $p \notin B^+(\U)$. By definition, this means there exists a point $q \in J^+(p) \cap \overline{\U}$ such that $q \neq p$.
    
    Because $q \in J^+(p)$, $q \neq p$, and $\tau$ is a time function, we must have $\tau(q) > \tau(p)$. This leads to a contradiction. On one hand,
    $$ \tau(q) > \tau(p) = \tau_{\max} $$
    On the other hand, since $q \in \overline{\U}$, its time value cannot exceed the maximum on that set, which implies $\tau(q) \le \tau_{\max}$.
    We have derived $\tau_{\max} < \tau(q) \le \tau_{\max}$, which is impossible. The assumption that $p \notin B^+(\U)$ must be false. Therefore, $p \in B^+(\U)$. This shows that $B^+(\U)$ is closed. As $B^+(\U)$ is a closed subset of the compact set $\overline{\U}$, it is itself compact. 
    \end{enumerate}
    The same proof for $B^-(\U)$ follows by replacing $J^+ \leftrightarrow J^-$.
\end{proof}

We also show that $B^\pm(\U)$ is covariant under isometries, which will relate to the covariance of the state update rule under measurements that we will introduce in the following section.

\begin{proposition}
    \label{prop: cov Bpm} Let $\varphi$ be any isometry on a globally hyperbolic spacetime $(\mathcal{M},g)$. Then for all precompact $\U \subset \mathcal{M}$, $B^\pm(\varphi(\U)) = \varphi(B^\pm(\U))$ respectively.
\end{proposition}

\begin{proof}
    Let us show that this holds for $B^+$ -- the proof for $B^-$ is completely analogous. We will prove this by showing both inclusions: $\varphi(B^+(\U)) \subseteq B^+(\varphi(\U))$ and $B^+(\varphi(\U)) \subseteq \varphi(B^+(\U))$. \\

Let $ p \in B^+(\mathcal{U}) $. Then $ p \in \overline{\mathcal{U}} $, so since $ \varphi $ is a homeomorphism, $ \varphi(p) \in \overline{\varphi(\mathcal{U})} $. Also, since $ J^+(p) \cap \overline{\mathcal{U}} = \{p\} $, and $ \varphi $ preserves causal relations and closures, we have 
$
J^+(\varphi(p)) = \varphi(J^+(p))$, $\overline{\varphi(\mathcal{U})} = \varphi(\overline{\mathcal{U}})$ so $J^+(\varphi(p)) \cap \overline{\varphi(\mathcal{U})} = \varphi(J^+(p)) \cap \varphi(\overline{\mathcal{U}}) = \varphi(J^+(p) \cap \overline{\mathcal{U}})$ which, for $p \in B^+(\U)$, is just $\{\varphi(p)\}$. Therefore, $ \varphi(p) \in B^+(\varphi(\mathcal{U})) $ so $\varphi(B^+(\mathcal{U})) = \{\varphi(p) \mid p \in B^+(\U)\} \subset B^+(\varphi(\mathcal{U}))$. \\

We now show the other inclusion. Let $ q \in B^+(\varphi(\mathcal{U})) $. Then $ q \in \overline{\varphi(\mathcal{U})} = \varphi(\overline{\mathcal{U}}) $, so there exists $ p \in \overline{\mathcal{U}} $ with $ \varphi(p) = q $. Also $J^+(q) \cap \overline{\varphi(\mathcal{U})} = \{q\}
\Rightarrow J^+(\varphi(p)) \cap \varphi(\overline{\mathcal{U}}) = \{\varphi(p)\}
\Rightarrow \varphi(J^+(p)) \cap \varphi(\overline{\mathcal{U}}) = \{\varphi(p)\}$. Applying $ \varphi^{-1} $, we get $J^+(p) \cap \overline{\mathcal{U}} = \{p\} \Rightarrow p \in B^+(\mathcal{U}) \Rightarrow q = \varphi(p) \in \varphi(B^+(\mathcal{U}))$. Therefore, $B^+(\varphi(\mathcal{U})) \subset \varphi(B^+(\mathcal{U}))$. Hence, $B^+(\varphi(\U)) = \varphi(B^+(\U))$. The same proof for $B^-(\U)$ follows by replacing $J^+ \leftrightarrow J^-$.
\end{proof}

We now show that there are always acausal Cauchy hypersurfaces that \enquote{intertwine} between spacelike-separated precompact regions.

\begin{lemma}
    \label{lem:conn spacelike hypersurfaces intertwining spacelike regions}
    For all precompact $\U,\V \subset \mathcal{M}$ such that $\overline{\U} \indep \overline{\V}$, there exists two acausal Cauchy hypersurfaces $\Sigma_1$ and $\Sigma_2$ such that
    \begin{enumerate}
        \item $\Sigma_1 \cap J^+(\U) \neq \varnothing$ and $\Sigma_1 \cap J^-(\U) = \varnothing$;
        \item $\Sigma_1 \cap J^+(\V) = \varnothing$ and $\Sigma_1 \cap J^-(\V) \neq \varnothing$;
        \item $\Sigma_2 \cap J^+(\U) = \varnothing$ and $\Sigma_2 \cap J^-(\U) \neq \varnothing$;
        \item $\Sigma_2 \cap J^+(\V) \neq \varnothing$ and $\Sigma_2 \cap J^-(\V) = \varnothing$.
    \end{enumerate}
\end{lemma}

\begin{proof}
    It is sufficient to show there there exists such a $\Sigma_1$ and apply a symmetric construction for $\Sigma_2$ with $\U \leftrightarrow \V$. The construction is displayed in Fig. \ref{fig:construction intertwining}.  
    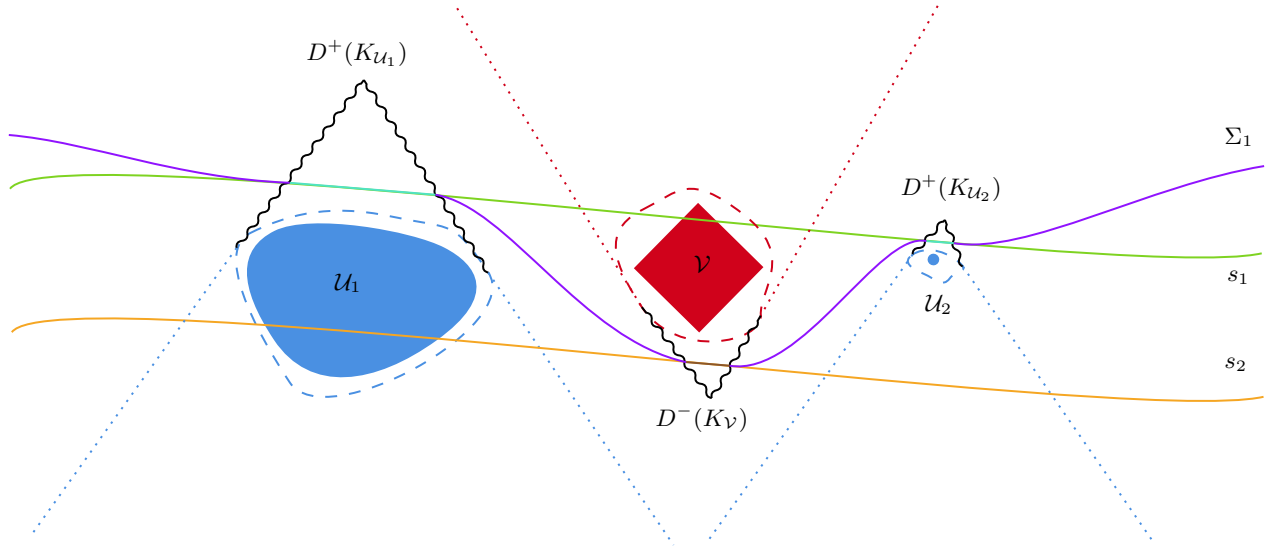
\begin{figure*}
        \centering
        \begin{tikzpicture}[x=0.75pt,y=0.75pt,yscale=-0.9,xscale=0.9]
\draw  [color={rgb, 255:red, 74; green, 144; blue, 226 }  ,draw opacity=1 ][fill={rgb, 255:red, 74; green, 144; blue, 226 }  ,fill opacity=1 ] (148.67,126.33) .. controls (168.67,116.33) and (225,126.33) .. (243,135.67) .. controls (261,145) and (273,161.67) .. (238.67,186.33) .. controls (204.33,211) and (168.67,216.33) .. (148.67,186.33) .. controls (128.67,156.33) and (128.67,136.33) .. (148.67,126.33) -- cycle ;
\draw  [color={rgb, 255:red, 74; green, 144; blue, 226 }  ,draw opacity=1 ][fill={rgb, 255:red, 74; green, 144; blue, 226 }  ,fill opacity=1 ] (513.29,142.28) .. controls (513.02,140.83) and (513.97,139.42) .. (515.42,139.15) .. controls (516.87,138.87) and (518.27,139.82) .. (518.55,141.28) .. controls (518.83,142.73) and (517.88,144.13) .. (516.42,144.41) .. controls (514.97,144.68) and (513.57,143.73) .. (513.29,142.28) -- cycle ;
\draw  [color={rgb, 255:red, 208; green, 2; blue, 27 }  ,draw opacity=1 ][fill={rgb, 255:red, 208; green, 2; blue, 27 }  ,fill opacity=1 ] (384.7,110.98) -- (420.35,146.04) -- (385.3,181.69) -- (349.65,146.63) -- cycle ;
\draw [color={rgb, 255:red, 126; green, 211; blue, 33 }  ,draw opacity=1 ]   (1,102.33) .. controls (25,69) and (621.67,155.67) .. (699,138.33) ;
\draw  [color={rgb, 255:red, 74; green, 144; blue, 226 }  ,draw opacity=1 ][dash pattern={on 4.5pt off 4.5pt}] (135,127) .. controls (155,117) and (178.33,110.33) .. (213,117.67) .. controls (247.67,125) and (247.67,119) .. (267.67,149) .. controls (287.67,179) and (177,240.33) .. (157,210.33) .. controls (137,180.33) and (115,137) .. (135,127) -- cycle ;
\draw [color={rgb, 255:red, 208; green, 2; blue, 27 }  ,draw opacity=1 ] [dash pattern={on 0.84pt off 2.51pt}]  (250.33,1.67) -- (353.33,168.33) ;
\draw [color={rgb, 255:red, 74; green, 144; blue, 226 }  ,draw opacity=1 ] [dash pattern={on 0.84pt off 2.51pt}]  (267.67,149) -- (371,301) ;
\draw  [color={rgb, 255:red, 208; green, 2; blue, 27 }  ,draw opacity=1 ][dash pattern={on 4.5pt off 4.5pt}] (354.33,113.67) .. controls (374.33,103.67) and (379.33,93) .. (409.67,116.33) .. controls (440,139.67) and (418.67,158.33) .. (419.67,173) .. controls (420.67,187.67) and (373.33,198.33) .. (353.33,168.33) .. controls (333.33,138.33) and (334.33,123.67) .. (354.33,113.67) -- cycle ;
\draw    (127.33,135.33) .. controls (127.01,133) and (128.02,131.67) .. (130.35,131.35) .. controls (132.68,131.02) and (133.69,129.69) .. (133.37,127.36) .. controls (133.05,125.03) and (134.06,123.7) .. (136.39,123.38) .. controls (138.72,123.05) and (139.73,121.72) .. (139.41,119.39) .. controls (139.1,117.06) and (140.11,115.73) .. (142.44,115.41) .. controls (144.77,115.09) and (145.78,113.76) .. (145.46,111.43) .. controls (145.14,109.1) and (146.15,107.77) .. (148.48,107.44) .. controls (150.81,107.12) and (151.82,105.79) .. (151.5,103.46) .. controls (151.18,101.13) and (152.19,99.8) .. (154.52,99.47) .. controls (156.85,99.15) and (157.86,97.82) .. (157.54,95.49) .. controls (157.22,93.16) and (158.23,91.83) .. (160.56,91.5) .. controls (162.89,91.18) and (163.9,89.85) .. (163.58,87.52) .. controls (163.26,85.19) and (164.27,83.86) .. (166.6,83.53) .. controls (168.93,83.21) and (169.94,81.88) .. (169.62,79.55) .. controls (169.3,77.22) and (170.31,75.89) .. (172.64,75.56) .. controls (174.97,75.24) and (175.98,73.91) .. (175.66,71.58) .. controls (175.34,69.25) and (176.35,67.92) .. (178.68,67.59) .. controls (181.01,67.27) and (182.02,65.94) .. (181.7,63.61) .. controls (181.38,61.28) and (182.39,59.95) .. (184.72,59.63) .. controls (187.05,59.3) and (188.06,57.97) .. (187.74,55.64) .. controls (187.42,53.31) and (188.43,51.98) .. (190.76,51.66) .. controls (193.09,51.33) and (194.1,50) .. (193.78,47.67) .. controls (193.46,45.34) and (194.47,44.01) .. (196.8,43.69) -- (198.33,41.67) -- (198.33,41.67) ;
\draw    (267.67,149) .. controls (265.36,148.51) and (264.46,147.11) .. (264.95,144.8) .. controls (265.45,142.49) and (264.55,141.09) .. (262.24,140.6) .. controls (259.93,140.11) and (259.03,138.71) .. (259.53,136.4) .. controls (260.02,134.09) and (259.12,132.69) .. (256.81,132.2) .. controls (254.5,131.71) and (253.6,130.31) .. (254.1,128) .. controls (254.6,125.69) and (253.7,124.29) .. (251.39,123.8) .. controls (249.08,123.31) and (248.18,121.91) .. (248.68,119.6) .. controls (249.17,117.29) and (248.27,115.89) .. (245.96,115.4) .. controls (243.65,114.91) and (242.75,113.51) .. (243.25,111.2) .. controls (243.75,108.89) and (242.85,107.49) .. (240.54,107) .. controls (238.23,106.51) and (237.33,105.11) .. (237.82,102.8) .. controls (238.32,100.49) and (237.42,99.09) .. (235.11,98.6) .. controls (232.8,98.11) and (231.9,96.71) .. (232.4,94.4) .. controls (232.89,92.09) and (231.99,90.69) .. (229.68,90.2) .. controls (227.37,89.71) and (226.47,88.31) .. (226.97,86) .. controls (227.47,83.69) and (226.57,82.29) .. (224.26,81.8) .. controls (221.95,81.31) and (221.05,79.91) .. (221.55,77.6) .. controls (222.04,75.29) and (221.14,73.89) .. (218.83,73.4) .. controls (216.52,72.91) and (215.62,71.51) .. (216.12,69.2) .. controls (216.62,66.89) and (215.72,65.49) .. (213.41,65) .. controls (211.1,64.51) and (210.2,63.11) .. (210.69,60.8) .. controls (211.19,58.49) and (210.29,57.09) .. (207.98,56.6) .. controls (205.67,56.11) and (204.77,54.71) .. (205.27,52.4) .. controls (205.76,50.09) and (204.86,48.69) .. (202.55,48.2) .. controls (200.24,47.71) and (199.34,46.31) .. (199.84,44) -- (198.33,41.67) -- (198.33,41.67) ;
\draw  [color={rgb, 255:red, 74; green, 144; blue, 226 }  ,draw opacity=1 ][dash pattern={on 4.5pt off 4.5pt}] (504.33,139) .. controls (510.33,135) and (520.33,135.67) .. (529,142.33) .. controls (537.67,149) and (529,147) .. (525,152.33) .. controls (521,157.67) and (517.33,152.67) .. (510.33,149.67) .. controls (503.33,146.67) and (498.33,143) .. (504.33,139) -- cycle ;
\draw    (504.33,139) .. controls (504.25,136.65) and (505.39,135.43) .. (507.74,135.34) .. controls (510.09,135.25) and (511.23,134.03) .. (511.15,131.68) .. controls (511.06,129.33) and (512.2,128.11) .. (514.55,128.02) .. controls (516.9,127.93) and (518.04,126.71) .. (517.96,124.36) .. controls (517.88,122.01) and (519.02,120.79) .. (521.37,120.7) -- (522.33,119.67) -- (522.33,119.67) ;
\draw    (532.33,145.67) .. controls (530.18,144.71) and (529.58,143.15) .. (530.54,141) .. controls (531.49,138.85) and (530.89,137.29) .. (528.74,136.33) .. controls (526.59,135.38) and (525.99,133.82) .. (526.95,131.67) .. controls (527.9,129.52) and (527.3,127.96) .. (525.15,127) .. controls (523,126.04) and (522.4,124.48) .. (523.36,122.33) -- (522.33,119.67) -- (522.33,119.67) ;
\draw [color={rgb, 255:red, 245; green, 166; blue, 35 }  ,draw opacity=1 ]   (2,182.33) .. controls (26,149) and (622.67,235.67) .. (700,218.33) ;
\draw    (392.33,219) .. controls (390,218.7) and (388.98,217.38) .. (389.28,215.04) .. controls (389.59,212.7) and (388.57,211.38) .. (386.23,211.08) .. controls (383.89,210.77) and (382.87,209.45) .. (383.18,207.11) .. controls (383.49,204.77) and (382.47,203.45) .. (380.13,203.15) .. controls (377.79,202.85) and (376.77,201.53) .. (377.08,199.19) .. controls (377.39,196.85) and (376.37,195.53) .. (374.03,195.23) .. controls (371.69,194.92) and (370.67,193.6) .. (370.98,191.26) .. controls (371.29,188.92) and (370.27,187.6) .. (367.93,187.3) .. controls (365.6,186.99) and (364.58,185.67) .. (364.89,183.34) .. controls (365.2,181) and (364.18,179.68) .. (361.84,179.38) .. controls (359.5,179.08) and (358.48,177.76) .. (358.79,175.42) .. controls (359.1,173.08) and (358.08,171.76) .. (355.74,171.45) -- (353.33,168.33) -- (353.33,168.33) ;
\draw    (392.33,219) .. controls (391.75,216.71) and (392.6,215.28) .. (394.89,214.7) .. controls (397.18,214.12) and (398.03,212.69) .. (397.44,210.4) .. controls (396.86,208.11) and (397.71,206.68) .. (400,206.1) .. controls (402.28,205.52) and (403.13,204.09) .. (402.55,201.81) .. controls (401.96,199.52) and (402.81,198.09) .. (405.1,197.51) .. controls (407.39,196.93) and (408.24,195.5) .. (407.66,193.21) .. controls (407.07,190.92) and (407.92,189.49) .. (410.21,188.91) .. controls (412.5,188.33) and (413.35,186.9) .. (412.77,184.61) .. controls (412.18,182.32) and (413.03,180.89) .. (415.32,180.31) .. controls (417.6,179.73) and (418.45,178.3) .. (417.87,176.02) -- (419.67,173) -- (419.67,173) ;
\draw [color={rgb, 255:red, 74; green, 144; blue, 226 }  ,draw opacity=1 ] [dash pattern={on 0.84pt off 2.51pt}]  (504.33,139) -- (389,300.33) ;
\draw [color={rgb, 255:red, 74; green, 144; blue, 226 }  ,draw opacity=1 ] [dash pattern={on 0.84pt off 2.51pt}]  (532.33,145.67) -- (638.33,299) ;
\draw [color={rgb, 255:red, 74; green, 144; blue, 226 }  ,draw opacity=1 ] [dash pattern={on 0.84pt off 2.51pt}]  (127.33,135.33) -- (12,296.67) ;
\draw [color={rgb, 255:red, 208; green, 2; blue, 27 }  ,draw opacity=1 ] [dash pattern={on 0.84pt off 2.51pt}]  (518.33,0.33) -- (419.67,173) ;
\draw [color={rgb, 255:red, 80; green, 227; blue, 194 }  ,draw opacity=1 ]   (157,99) -- (238.33,105.67) ;
\draw [color={rgb, 255:red, 80; green, 227; blue, 194 }  ,draw opacity=1 ]   (512,131.5) -- (527.33,132.67) ;
\draw [color={rgb, 255:red, 139; green, 87; blue, 42 }  ,draw opacity=1 ]   (378,199) -- (402.33,201) ;
\draw [color={rgb, 255:red, 144; green, 19; blue, 254 }  ,draw opacity=1 ]   (0.33,72.33) .. controls (46.33,76.33) and (92.33,96.33) .. (157,99) ;
\draw [color={rgb, 255:red, 144; green, 19; blue, 254 }  ,draw opacity=1 ]   (238.33,105.67) .. controls (279,112.33) and (319.67,186.33) .. (378,199) ;
\draw [color={rgb, 255:red, 144; green, 19; blue, 254 }  ,draw opacity=1 ]   (527.33,132.67) .. controls (568,139.33) and (635.67,100.33) .. (700.33,89.67) ;
\draw [color={rgb, 255:red, 144; green, 19; blue, 254 }  ,draw opacity=1 ]   (402.33,201) .. controls (443,207.67) and (485.67,125.67) .. (512,131.5) ;
\draw (678,146.4) node [anchor=north west][inner sep=0.75pt]    {$s_{1}$};
\draw (676.67,195.07) node [anchor=north west][inner sep=0.75pt]    {$s_{2}$};
\draw (180.67,147.73) node [anchor=north west][inner sep=0.75pt]    {$\U_{1}$};
\draw (510,159.73) node [anchor=north west][inner sep=0.75pt]    {$\U_{2}$};
\draw (380.67,137.73) node [anchor=north west][inner sep=0.75pt]    {$\V$};
\draw (164.67,17.73) node [anchor=north west][inner sep=0.75pt]    {$D^{+}( K_{\U_{1}})$};
\draw (496,92.73) node [anchor=north west][inner sep=0.75pt]    {$D^{+}( K_{\U_{2}})$};
\draw (359.33,221.4) node [anchor=north west][inner sep=0.75pt]    {$D^{-}( K_{\V})$};
\draw (676.67,66.4) node [anchor=north west][inner sep=0.75pt]    {$\Sigma _{1}$};
\end{tikzpicture}
        \caption{Construction for intertwining an acausal Cauchy hypersurface between two precompact spacelike-separated regions. Here, we have that $\U = \U_1 \cup \U_2$ (blue filled) where $\U_2$ is a point, and $\V \indep \U$ where $\V$ (red filled) is a causal diamond. Although a causal diamond has an empty domain of dependence, the neighbourhood $K_\V$ (dashed red) has a nontrivial past domain of dependence; likewise, although a point has an empty domain of dependence, the neighbourhood of $\U_2$ has a nontrivial future domain of dependence. Cauchy hypersurfaces $s_1$ (green) and $s_2$ (orange) intersect $D^+(K_{\U_1}) \cup D^+(K_{\U_2})$ and $D^-(K_\V)$ at $\sigma_1$ (turquoise) and $\sigma_2$ (brown), respectively. The compact acausal surface $H = \sigma_1 \cup \sigma_2$ can be used to create a Cauchy spacelike hypersurface $\Sigma_1$ (purple) intertwining $\U$ and $\V$.}
        \label{fig:construction intertwining}
    \end{figure*}
    Since $\U$ and $\V$ are precompact such that $\overline{\U} \indep \overline{\V}$, there exists compact neighbourhoods $K_\U \supset \U$ and $K_\V \supset \V$ such that $K_\U \indep K_\V$ and $D^+(K_\U) \neq \varnothing$ and $D^-(K_\V) \neq \varnothing$. In addition, there exists two acausal Cauchy hypersurfaces $s_1$ and $s_2$ such that $\sigma_1 := s_1 \cap D^+(K_\U) \neq \varnothing$ and $\sigma_2 := s_2 \cap D^-(K_\V) \neq \varnothing$ (this can be done using lemma \ref{lem:Cauchy in the future} by taking $\tau = c$ with $m_\text{max} > c > m_0$ since $\U$ and $\V$ are spacelike-separated and precompact, for a sufficiently small $m_\text{max}$). Then $H = \sigma_1 \cup \sigma_2$ is an acausal (as $\sigma_1$ and $\sigma_2$ are spacelike-separated since $K_\U \indep K_\V \Rightarrow D^+(K_\U) \indep D^-(K_\V)$ and $s_1 \cap D^+(K_\U) \subseteq D^+(K_\U) \indep D^-(K_\V) \supseteq s_2 \cap D^-(K_\V)$, and these are individually acausal as subsets of acausal Cauchy hypersurfaces) compact region, so by theorem \ref{Thm: three cauchy hypersurfaces bernal} there exists an acausal Cauchy hypersurface with the required conditions.
\end{proof}

It is thus possible to intertwine acausal Cauchy hypersurfaces in between such spacelike-separated precompact regions of globally-hyperbolic spacetimes. This, alongside theorem \ref{thm:muller}, implies the following result.

\begin{corollary}
    \label{cor:Cauchy functions Sigma-Sigma1 and Sigma-Sigma2}
    There exists two Cauchy time functions $\tau_1 : \mathcal{M} \to \mathbb{R}$ and $\tau_2 : \mathcal{M} \to \mathbb{R}$ such that
    \begin{enumerate}
        \item $\Sigma_i = \tau_1^{-1}(t_i)$ and $\Sigma_1 = \tau_1^{-1}(t)$ and $\Sigma_f = \tau^{-1}(t_f)$ for some $t_i < t < t_f \in \mathbb{R}$,
        \item $\Sigma_i = \tau_2^{-1}(\tilde{t}_i)$ and $\Sigma_2 = \tau_2^{-1}(\tilde{t})$ and $\Sigma_f = \tau^{-1}(\tilde{t}_f)$ for some $\tilde{t}_i < \tilde{t} < \tilde{t}_f \in \mathbb{R}$,
    \end{enumerate}
    where $\Sigma_1$ and $\Sigma_2$ are given as in lemma \ref{lem:conn spacelike hypersurfaces intertwining spacelike regions} and are such that $\Sigma_1,\Sigma_2 \subset J^+(\Sigma_i)$ and $\Sigma_1,\Sigma_2 \subset J^-(\Sigma_f)$ for $\Sigma_i,\Sigma_f$ two spacelike Cauchy hypersurfaces in the strict causal past and causal future of both $\U$ and $\V$, respectively.
\end{corollary}

There are thus two families of Cauchy hypersurfaces which both include $\Sigma_i$ and $\Sigma_f$ and which intertwine precompact spacelike-separated regions in a different order. The following theorem then follows as a corollary by choosing a sufficiently large $m_0$ for $\Sigma_i$ and $\Sigma_f$ in lemma \ref{lem:Cauchy in the future} such that $\Sigma_1,\Sigma_2 \subset I^+(\Sigma_i)$ and $\Sigma_1,\Sigma_2 \subset I^-(\Sigma_f)$.

\begin{theorem}
    \label{thm:Foliation Cauchy}
    Let $\U$ and $\V$ be two precompact regions of a globally hyperbolic spacetime $(\mathcal{M},g)$ such that $\overline{\U} \indep \overline{\V}$. Then there exists foliations of $\mathcal{M}$ by means of two families $\mathcal{S}_1$ and $\mathcal{S}_2$ of Cauchy hypersurfaces such that 
    \begin{enumerate}
        \item there exists an acausal Cauchy hypersurface $\Sigma_1 \in \mathcal{S}_1$ such that $\Sigma_1 \cap J^+(\U) \neq \varnothing$ and $\Sigma_1 \cap J^-(\U) = \varnothing$ while $\Sigma_1 \cap J^+(\V) = \varnothing$ and $\Sigma_1 \cap J^-(\V) \neq \varnothing$, i.e. $\Sigma_1$ is in the strict causal future of $\U$ and strict causal past of $\V$;
        \item there exists an acausal Cauchy hypersurface $\Sigma_2 \in \mathcal{S}_2$ such that $\Sigma_2 \cap J^+(\U) = \varnothing$ and $\Sigma_2 \cap J^-(\U) \neq \varnothing$ and $\Sigma_2 \cap J^+(\V) \neq \varnothing$ and $\Sigma_2 \cap J^-(\V) = \varnothing$, i.e. $\Sigma_2$ is in the strict causal past of $\U$ and strict causal future of $\V$.
        \item there exists two acausal Cauchy hypersurfaces $\Sigma_i,\Sigma_f \in \mathcal{S}_1 \cap \mathcal{S}_2$ such that $\Sigma_1,\Sigma_2 \subset J^+(\Sigma_i) \cap J^-(\Sigma_f)$, i.e. $\Sigma_i$ and $\Sigma_f$ lie, respectively, in the causal past and future of both $\Sigma_1$ and $\Sigma_2$.
    \end{enumerate}
\end{theorem}

\end{appendix}

\onecolumngrid

\bibliographystyle{unsrt}
\bibliography{references,references2}

\end{document}